\documentclass{article}
\usepackage{library}

\usepackage{amsthm}
\usepackage{thmtools,thm-restate}

\numberwithin{equation}{section}
\declaretheoremstyle[bodyfont=\it,qed=\qedsymbol]{noproofstyle}

\declaretheorem[name=Observation,numbered=no]{observation*}

\declaretheorem[numberlike=equation]{theorem}

\declaretheorem[name=Theorem,numbered=no]{theorem*}

\declaretheorem[numberlike=equation]{lemma}
\declaretheorem[name=Lemma,numbered=no]{lemma*}

\declaretheorem[numberlike=equation]{corollary}
\declaretheorem[name=Corollary,numbered=no]{corollary*}

\declaretheorem[numberlike=equation]{proposition}
\declaretheorem[name=Proposition,numbered=no]{proposition*}

\declaretheorem[numberlike=equation]{claim}
\declaretheorem[name=Claim,numbered=no]{claim*}

\declaretheorem[name=Conjecture,numbered=no]{conjecture*}

\declaretheorem[numberlike=equation]{question}
\declaretheorem[name=Question,numbered=no]{question*}

\declaretheoremstyle[bodyfont=\it,qed=$\lrcorner$]{defstyle} 

\declaretheorem[numberlike=equation,style=defstyle]{definition}
\declaretheorem[unnumbered,name=Definition,style=defstyle]{definition*}

\declaretheorem[numberlike=equation,style=defstyle]{example}
\declaretheorem[unnumbered,name=Example,style=defstyle]{example*}

\declaretheorem[unnumbered,name=Notation=defstyle]{notation*}

\declaretheorem[unnumbered,name=Construction,style=defstyle]{construction*}

\declaretheorem[numberlike=equation,style=defstyle]{remark}
\declaretheorem[unnumbered,name=Remark,style=defstyle]{remark*}

\newcommand{\N}{\mathbb N}

\newcommand{\C}{\mathbb C}

\newcommand{\F}{\mathbb F}

\newcommand{\cM}{{\mathcal{M}}}   
\newcommand{\cO}{{\mathcal{O}}}   
\newcommand{\cQ}{{\mathcal{Q}}}

\newcommand{\abs}[1]{\left| {#1} \right|}

\renewcommand{\epsilon}{\varepsilon}

\renewcommand{\tilde}{\widetilde}

\newif\ifdraft
\drafttrue
\ifdraft
\newcommand{\ASnote}[1]{\textcolor{Blue}{\guillemotleft AS: #1 \guillemotright}}

\else
\newcommand{\ASnote}[1]{}
\fi

\newcommand{\NL}{\mathsf{nl}}
\newcommand{\Ind}{\mathsf{Ind}}
\newcommand{\agree}{\mathsf{agree}}
\newcommand{\Tr}{\mathsf{Tr}}
\newcommand{\maxagree}{\mathsf{maxagree}}
\newcommand{\supp}{\mathsf{supp}}
\newcommand{\prob}{\mathbb{P}}
\newcommand{\expectation}{\mathbb{E}}
\newcommand{\poly}{\operatorname{poly}}
\newcommand{\turan}{\text{Tur\'an}}
\newcommand{\indicator}{\textbf{1}}
\newcommand{\linagree}[1]{\maxagree(#1, \textnormal{deg }1)}
\newcommand{\magree}[2]{\maxagree(#1, \textnormal{deg }#2)}

\newcommand{\ignore}[1]{}
\DeclareMathOperator{\Char}{char}

\newcommand{\KK}[2]{K^{(#2)}_{#1}}

\newif\ifblind
\blindfalse   

\ifblind
\newcommand{\skc}[1]{}
\newcommand{\src}[1]{}
\newcommand{\rkc}[1]{}
\else
\newcommand{\skc}[1]{\textcolor{red}{SK: #1}}
\newcommand{\src}[1]{\textcolor{olive}{SR: #1}}
\newcommand{\rkc}[1]{\textcolor{violet}{RK: #1}}
\fi

\title{Explicit Nonlinear Functions beyond the Fourier bound}
\author{Swastik Kopparty\thanks{Department of Mathematics and Department of Computer Science, University of Toronto. Email: \texttt{swastik.kopparty@utoronto.ca}.   Research partially supported
by an NSERC Discovery Grant.} \and Rishabh Kothary\thanks{Department of Computer Science, University of Toronto.  Email: \texttt{rishabh.kothary@mail.utoronto.ca}.  Research partially supported by an NSERC Discovery Grant.} \and Shanthanu S. Rai\thanks{School of Technology and Computer Science, Tata Institute of Fundamental Research, Mumbai. Email: \texttt{shanthanu.rai@tifr.res.in}.  Research supported by the Department of Atomic Energy, Government of India, under project number RTI4014, TCS Research Fellowship and in part by research grants from Google Research, Premji Invest and SERB.}}
\date{}

\begin{document}
\maketitle

\begin{abstract}
    We study the problem of constructing highly nonlinear vectorial maps $F: \F_2^n \to \F_2^m$.
    Concretely, we want an $F$ and an $A = A(m,n)> 0$ as small as possible, so that for every affine map $L: \F_2^n \to \F_2^m$ (of the form $L(x) = M x + b $) we have:
    $$\agree(F, L) :=  |\{ x \in \F_2^n \mid  F(x) = L(x) \}| \leq A.$$
    Such questions have been studied by Nyberg~\cite{Nyberg91,Nyberg}, Carlet and Ding~\cite{CD04,CD07}, Liu, Mesnager and Chen~\cite{LMC17}, Nagy~\cite{Na25},
    and Biryukov, Turecek, and Udovenko~\cite{BTU26}.
    There is a classical method of constructing such functions from bent-functions and Fourier analytic ideas; the best bound achievable by this method is:
    $$ A(m,n) = \Theta(2^{n-m} + 2^{n/2}),$$
    and in particular, is never smaller than $2^{n/2}$.

    In this work, we show how to construct highly nonlinear functions beyond this Fourier bound. Concretely, we show how to construct for every $\gamma>0$, a function 
    $F: \F_2^n \to \F_2^m$ with $m = \cO_{\gamma}(n)$, achieving
    $$ A(m,n) \leq  (1 + \gamma)^n.$$

    Surprisingly, we even achieve the same quantitative behavior for the much harder question of having low agreement with $m$-tuples of degree $d$ polynomials $Q: \F_2^n \to \F_2^m$, 
    with $m = \cO_{\gamma, d}(n)$. Here the previously best bounds were of the form $A(m,n) = \cO( 2^{-\frac{n}{2^{d+1}}}  \cdot 2^n )$ of Ben-Sasson and Kopparty~\cite{KoppartyThesis}, based on Gowers-norm-type arguments.
    
    All our results generalize to all finite fields $\F_q$ in place of $\F_2$.
    
    Our methods are based on a new connection to classical results on counting solutions to systems of polynomial equations via algebraic methods. This connection brings us to basic questions in combinatorics, about graphs and hypergraphs with simultaneously a small number of edges and independent sets.
\end{abstract}

\section{Introduction}

This paper is about {\em highly nonlinear} functions $F: \F_2^n \to \F_2^m$,
and their degree $d$ analogues. Our main result is the first explicit construction of functions whose nonlinearity goes beyond the classical Fourier bound, and corresponding results for the degree $d$ analogues. Our results are obtained by a new connection of this problem to (1) the basic algebraic problem of bounding the number of solutions to systems of polynomial equations, and (2) the basic combinatorial problem of bounding the number of independent sets in graphs or hypergraphs with a given number of edges.

A vectorial Boolean function is a function $F: \F_2^n \to \F_2^m$.
Linearity properties of vectorial Boolean functions have been extensively studied in differential cryptanalysis and symmetric-key cryptography (where vectorial Boolean functions are called {\em substitution boxes} or {\em $S$-boxes}).

An important parameter of a vectorial Boolean function $F$ is its {\em nonlinearity}\footnote{This specific definition is more commonly called nonlinearity of the second type, or vector nonlinearity~\cite{Na25}.} $\NL(F)$. This is defined as the minimum Hamming distance of $F$ from affine maps $L : \F_2^n \to \F_2^m$ (namely, maps where $L(x) = Mx + b$ for a matrix $M$ and vector $b$):
$$ \NL(F) =   \min_{\substack{L: \F_2^n \to \F_2^m\\\mbox{ affine map}}}  |\{ x \in \F_2^n \mid  F(x) \neq L(x) \}|.$$
For our discussion, it will be more natural to talk about the ``linearity'', denoted $ \linagree{F}$ defined to equal $2^n - \NL(F)$, which is equivalently the smallest $A \geq 0$ such that for all affine $L : \F_2^n \to \F_2^m$:
$$ |\{ x \in \F_2^n \mid  F(x) = L(x) \}| \leq A.$$
For cryptography and pseudorandomness purposes, it is interesting to find $F$ which have provably small $\linagree{F}$. This (and its higher degree analogue) is the main question that our paper focuses on.

Nonlinearity of vectorial Boolean functions has been studied by many authors~\cite{Nyberg91,Nyberg,CD04,CD07,chen2001nonlinearity,liu2014nonlinearity,LMC17,Na25,BTU26}. In the very recent work of Biryukov, Turecek, and Udovenko~\cite{BTU26}, nonlinearity plays a central role in their study of linearization of S-boxes.

For the case $m = 1$ (namely the case of Boolean functions), 
it is well known that the smallest that $\linagree{F}$ can be is 
$\frac{1}{2} \cdot 2^n + O(2^{n/2})$. That it cannot be smaller is proved by studying the Fourier transform of 
$F$ and applying Parseval's inequality. Functions which achieve equality are known as (almost) bent functions, and have been very well studied 
with numerous applications.

For general $m > 1$, a simple application of the probabilistic method
shows that there are $F: \F_2^n \to \F_2^m$ with 
$$\linagree{F} \leq \begin{cases} (1 + o(1)) \cdot\frac{1}{2^m} \cdot  2^n & m \ll n \\ \poly(n)   & m \geq n \end{cases}.$$
and this is also known to be best possible.

In particular, once $m \geq n$, there exist $F$ with $\linagree{F} \leq \poly(n)$.

However, constructions of $F$ with this behaviour are not known.
There is a very natural and elegant approach to constructing nonlinear vectorial Boolean functions based on bent functions and
Fourier analysis. This dates back to Nyberg~\cite{Nyberg91},

with more recent work by \cite{CD04,CD07,chen2001nonlinearity,liu2014nonlinearity,LMC17,BTU26,Na25}.
Specifically, suppose $F: \F_2^n \to \F_2^m$ is such that for every
nonzero $y \in \F_2^m$, the one-bit-output function $F_y: \F_2^n \to \F_2$  given by:
$$ F_y(x) = \langle F(x), y \rangle$$
is almost bent; namely $\linagree{F_y} \leq \frac{1}{2} \cdot 2^n + O(2^{n/2})$.
Then,  via the results in these papers, we have:

$$\linagree{F} \leq \frac{1}{2^m} \cdot 2^n + O(2^{n/2}).$$
We call this the {\em Fourier bound} on $\linagree{F}$. It is the best possible bound on $\linagree{F}$ based solely on the smallness of all the $\linagree{F_y}$.

For $m \ll n$, this is quite close to best possible, but for $m \geq n$ it is significantly suboptimal. The best upper bound on $\linagree{F}$ that can be proved via this method is 
$2^{n/2}$. There are no better upper bounds known for any explicit function, and it is of great interest to develop new techniques to prove strong upper bounds on $\linagree{F}$.

Our main result for degree $1$ (see \cref{thm:explicit-non-linear-functions-beyond-fourier-bound}) gives explicit functions which provably go beyond the Fourier bound.
Specifically, for every fixed $\gamma > 0$, for large enough $n$ and $m \geq \Omega_\gamma(n)$, 
we give explicit functions $F: \F_2^n \to \F_2^m$ with 
$$\linagree{F} \leq (1+ \gamma)^n.$$

One can ask similar questions for higher degree $d = O(1)$. Here we replace affine maps $L: \F_2^n \to \F_2^m$ with $m$-tuples of degree $d$ polynomials. 
Specifically, define $\magree{F}{d}$ to be the smallest $A$ such that for all $m$-tuples $Q = (Q_1, \ldots, Q_m)$ of polynomials (where $Q_i \in \F_2[X_1, \ldots, X_n]$, with $\deg(Q_i) \leq d$), we have:
$$ | \{ x \in \F_2^n \mid F(x) = Q(x) \} | \leq A.$$

For $m \ll n$, it is again known that the smallest possible value of $\magree{F}{d}$ is about $\left(\frac{1}{2^m}  + o(1)\right) \cdot 2^n$, and explicit constructions of such functions are known from multiple routes: (1) the Weil bounds and the Fourier bound, or (2) the direct product theorem of Viola and Wigderson~\cite{ViolaWigderson}. 
However for $m \geq n$, the known results get much weaker. Using ideas related to the Gowers' norm
(which is a higher-degree version of Fourier analysis), Ben-Sasson and Kopparty~\cite[Theorem 7.2.1]{KoppartyThesis} gave constructions of explicit functions $F$ satisfying:
$$ \magree{F}{d} \leq 2^{-n/2^{d+1}} \cdot 2^n.$$
As far as we could find, this is the best upper bound known for any $m = O(n)$ and $d \geq 2$.

Our main result for degree $d$ (see \cref{thm:explicit-non-linear-functions-beyond-fourier-bound-arbitrary-d}) gives explicit functions significantly improving this. In fact, our bounds for general (constant) $d$ even bypass the Fourier bound from the $d = 1$ case! Concretely, for every fixed $\gamma > 0$ and $d \in \mathbb N$, for large enough $n$ and $m \geq \Omega_{\gamma, d}(n)$, we give explicit functions $F: \F_2^n \to \F_2^m$ with:
$$\magree{F}{d} \leq (1+ \gamma)^n.$$

Our results also generalize to $\F_q$ in place of $\F_2$, achieving the same qualitative behaviour of going below the Fourier bound (see \cref{theorem:disjoint-clique-NH-count-Fq}).

\subsection{Techniques and Related Work}

Our results are based on a combination of algebraic and combinatorial tools.
The explicit functions $F$ we construct will be computed by carefully chosen polynomial maps.
The problem of bounding the agreement of these explicit functions with $m$-tuples of degree-$d$ polynomials $Q: \F_2^n \to \F_2^m$ then becomes a question of bounding the number of solutions to a 
system of polynomial equations $F- Q = 0$. We thus need (1) a method for bounding the number of solutions to a system of polynomial equations, and (2) a choice of $F$ so that this method 
will work on $F-Q$ {\em for all $Q$}.

The method in (1) will be related to {\em Macaulay's theorem}, which is a beautiful linear algebraic technique: it bounds the size of the set $S$ of solutions to a system of equations by examining the vector space $W$ of polynomial functions on $S$, and finding a small basis $B$ for $W$ using the polynomial equations defining $S$. This gives us an upper bound $|S| = \dim(W) \leq |B|$, where the first equality is because every function on $S$ is a polynomial function.

To make this work in our setting, we need to choose an $F$ appropriately. We show how to choose an $F$ based on any $(d+1)$-uniform hypergraph $H = (V,E)$, with $|V| = n$ and $|E| = m$.
Each edge $E$ specifies a monomial in the input variables $X_1, \ldots, X_n$. 
The bound on the number of solutions obtained by the Macaulay theorem argument turns out to be precisely $|\Ind(H)|$, the number of independent sets in $H$. This brings us to an extremely clean and concrete combinatorial question about hypergraphs: how do we design $(d+1)$-uniform hypergraphs on $n$ vertices with $m = O(n)$ edges so that $H$ has few independent sets.

Already simple instantiations of this framework, taking $H$ to be a disjoint union of constant-size complete hypergraphs, yield the main qualitative aspects of our results. For example, when $d = 1$, for every integer $k$
we get explicit functions $F : \F_2^n \to \F_2^m$ 
with $m = \frac{k-1}{2} n$ and
$$\linagree{F}\leq (k+1)^{n/k},$$
which is at most $(1 + \gamma)^n$ for large enough $k = O_{\gamma}(1)$.
In particular, this goes beyond the Fourier bound for any $k \geq 6$.

For the case of $\F_q$ in place of $\F_2$, our techniques generalize with some  modifications. Here some additional complexities are introduced by the presence of non-multilinear monomials. The resulting combinatorial problem that needs to be solved then becomes a bit more complicated, and one has to study \emph{multiset} analogues of edges vs a \emph{multiset} analogue of the independent sets question for graphs and hypergraphs. 

Interestingly, the explicit functions that arise from this framework turn out to have a {\em perfect} direct product theorem (see \cref{lemma:direct-product-F2}) for their agreement with degree $d$ polynomials; this is in contrast to the best known general direct-product theorem of Viola and Wigderson~\cite{ViolaWigderson}, which is lossy.

We also remark that the Macaulay theorem bound is also known as the {\em footprint bound}, and has a number of interesting applications in coding theory and theoretical computer science~\cite{GH00,GMT08,STV20,ABPSS25}.

Another strong motivation for this work comes from the beautiful matrix construction problem of Aggarwal, Dutta, Mukherjee, Nagargoje and Obremski~\cite{ADMNO25}, as a route to constructing $2$-source extractors. 
It seems like highly-nonlinear functions are very relevant to that problem, which also has a Fourier bound barrier.

\subsection{Refinements}

We now discuss the finer quantitative questions arising from our framework.
\begin{itemize}
    \item For $d = 1$ (namely for constructing highly nonlinear functions), our framework asks to find
{\em graphs} with $m$ edges minimizing the number of independent sets. This question was solved by a difficult theorem of Sah, Sawhney, Stoner and Zhao~\cite{SAH2019172}, who showed that the disjoint union of constant sized cliques is in fact the optimal choice.

\item For larger $d$, our framework asks to find $(d+1)$-uniform hypergraphs with $m$ edges minimizing the number of independent sets. It turns out that the disjoint union of constant size complete hypergraphs is {\em not} the optimal choice. For sufficiently large fixed $d$, one can use the {\em $\turan$ hypergraphs} (which are conjectured extremal examples for the hypergraph $\turan$ problem) to beat those constructions.

\item It turns out that the $\turan$ hypergraphs are not always optimal either. We know about an ad hoc construction of an $8$-vertex $3$-uniform hypergraph\footnote{See \cref{def:8-vertex-H*} for the definition of $H^*$} $H^*$ such that disjoint unions of $H^*$ beat every other construction that we know for $d = 2$ and $m = n$. 

\end{itemize}
To summarize the status of the finer quantitative questions: for $d =1$ we understand the limits of our framework, but for $d \geq 2$ there are very interesting combinatorial questions that arise. \cref{table:F_2-constructions} summarizes and compares the different explicit constructions of $F:\F_2^n \longrightarrow\F_2^m$.

\begin{table}[h!]
    \centering
\begin{tabular}{||m{4.0cm}|m{1.5cm}|m{4.0cm}|m{3.0cm}||}
     \hline
     \textbf{Construction} & \textbf{Degree} $d$ & \textbf{Output ratio } $m/n$ & \textbf{Agreement upper bound }  \\
     \hline \hline
     Fourier Bound & $d = 1$ &  $\geq 1$ & $\cO(2^{n/2})$ \\
     \hline 
     Ben-Sasson, Kopparty \cite{KoppartyThesis} & Arbitrary $d$ &  $\geq 1$ & $\cO\left( 2^{\left(1- \frac{1}{2^{d+1}}\right)\cdot n} \right)$\\
     \hline
     Disjoint Union of Cliques \cref{thm:explicit-non-linear-functions-beyond-fourier-bound} & $d=1$ & $\frac{k-1}{2}$ for arbitrary $k$ & $(k+1)^{n/k}$\\
     \hline
     Disjoint Union of Hypercliques \cref{thm:explicit-non-linear-functions-beyond-fourier-bound-arbitrary-d} & Arbitrary $d \geq 2$ & $\frac{1}{k}\cdot \binom{k}{d+1}$ for arbitrary $k$ & $\cO(k^{nd/k})$ \\
     \hline
     Disjoint Union of $\turan$ blocks \cref{corollary:Disjoint-Union-Turan-Systems} & Arbitrary $d$ & $\leq\frac{4}{k(d+3)}\cdot \binom{k}{d+1}$ for large enough $k$ & $\cO(k^{n(d+1)/k})$\\ 
     \hline
     Disjoint Union of $H^*$ \cref{lemma:sparse-hypergraph-many-ind-sets} & $d = 2$ & =1 & $115^{n/8}$\\
     \hline
\end{tabular}
\caption{Comparison of different constructions of $F :\F_2^{n} \longrightarrow \F_2^m$  }
\label{table:F_2-constructions}
\end{table}

\paragraph{AI use:} Our framework for constructing explicit functions was developed without any AI involvement. The writing was done without any AI involvement. We used AI to search for fixed size hypergraphs with a given number of edges that optimized the number of independent sets. AI found for us the 3-uniform Turan hypergraph and gave the improved ad-hoc example with $d = 2$. We used AI for literature search, and for comments on our writeup.

\section*{Acknowledgements}
We are very grateful to Divesh Aggarwal and Maciej Obremski for telling us about~\cite{ADMNO25}, for valuable discussions, and for encouraging us to consider the $m>n$ regime (which is the setting for our strongest results).

\paragraph{Organization of this paper:} We begin in \cref{sec:prelims} by going over some useful preliminaries about polynomials, hypergraphs and Fourier analysis. Then, in \cref{sec:non-linearity}, we formally define the notion of nonlinearity for vectorial boolean functions and show the existence of highly nonlinear functions. We go over the existing approach of constructing explicit nonlinear functions using Fourier analysis and comment on the limitations of the ``basic Fourier bound''. In \cref{sec:hard-functions-over-F2}, we introduce the \emph{hypergraph framework} for constructing nonlinear and ``far from degree $d$'' functions over $\F_2$. In \cref{sec:hard-functions-over-Fq},
we generalize this to do the same over $\F_q$. The presentation turns out to be cleaner for $\F_2$, wherein we show that the maximum agreement between $f_H$ and tuples of polynomials $\bar Q = (Q_1, \dots, Q_m)$ of degree at most $d$ is equal to the number of independent sets in $H$. After describing the hypergraph framework, we instantiate the framework with explicit hypergraphs to construct explicit nonlinear and ``far from degree $d$'' functions which beat the Fourier bound.

\section{Preliminaries}
\label{sec:prelims}

\begin{itemize}
    \item For a prime power $q$, $\F_q$ denotes the finite field of size $q$.
    \item $[n]$ denotes the set $\{ 1, 2, \dots, n\}$.
    \item We use $\bar X$ to denote the tuple $(X_1,X_2,\ldots,X_n)$, where the number of variables $n$ will be clear from context or specified when necessary.
\end{itemize}

\subsection{Polynomials and Functions}

\begin{definition}
    For any $S\subset \F_q^n$, let $\F_q^{S}$ denote the $\F_q$-vector space of functions from $S$ to $\F_q$.
\end{definition}
Note that $\dim_{\F_q}(\F_q^S) = |S|$.

\begin{definition}[Support of Polynomials]
    Let $f \in \F_q[\bar X]$ such that 
    \begin{equation*}
        f = \sum_{i \in \N^n}c_i \bar X^{i}
    \end{equation*}
    then we define
    \begin{equation*}
        \supp(f) = \{\bar X^{i} \mid c_i \neq 0\}
    \end{equation*}
    which we call the support of $f$.
\end{definition}

We denote the set of multilinear monomials of $n$ variables by $\cM_{mult,n}$. There is a natural correspondence between monomials $M \in \cM_{mult,n}$ and the subsets of $[n]$. In particular for any $W \subset [n]$, the monomial corresponding to it is $\bar X^{\textbf{1}_{W}}$ where $\textbf{1}_W \in \{0,1\}^{n}$ is the characteristic vector of $W$. Whenever $n$ is clear from context, we will call $\cM_{mult,n}$ by $\cM_{mult}$ for ease of notation.\\

We will require the following simple form of the Combinatorial Nullstellensatz \cite{ALON_1999}.

\begin{lemma}\cite[Lemma 2.1]{ALON_1999}\label{lemma:non-root-ind-deg-bound}
    Let $P \in \F[X_1,...,X_n]$ be a non-zero polynomial such that for $i \in [n]$, $\deg_{X_i}(P) \leq t_i$. Let $S_i \subseteq \F$, such that $|S_i| \geq t_i+1$ for $i \in [n]$, then there exists $(\alpha_1,...,\alpha_n) \in S_1 \times \dots \times S_n$ such that $P(\alpha_1,...,\alpha_n) \neq 0$.
\end{lemma}

\subsection{Hypergraphs}

An $r$-uniform hypergraph $H = (V,E)$ is a set of vertices $V$, along with a set $E$ of $r$-element subsets of $V$. The elements of $E$ are called the (hyper)edges of $H$. For the hypergraph $H$, we use the notation $V(H)$ to denote its vertex set, and $E(H)$ to denote its (hyper)edge set.

We next define the $k$-shadow of a $r$-uniform hypergraph, where $k \leq r$ as follows.

\begin{definition}[$k$-shadow of $r$-uniform hypergraph]
    Let $H = (V,E)$ be an $r$-uniform hypergraph. Then for any $1 \leq k \leq r$, we define the $k$-shadow of $H$ as the following $k$-graph $\partial_{k}H = (V, E_k)$ where
    \begin{equation*}
        E_{k} = \{ \{i_1,..,i_k\} \subset V \mid \exists i_{k+1},...,i_r \in V \text{ such that }\{ i_1,...,i_r \} \in E\}.
    \end{equation*}
\end{definition}

We next define independent sets of an $r$-uniform hypergraph $H$.
\begin{definition}
    Let $H = (V,E)$ be an $r$-uniform hypergraph.
    A subset $V' \subseteq V$ is an {\em independent set} in $H$ if
    there is no edge $e \in E$ such that $ e\subseteq V'$. 

    For $k \in \N$, we define the set of all independent sets of $H$ of size $k$ as
    \begin{equation*}
        \Ind(H,k) = \{ V' \subset V \mid V' \text{ is an independent set in $H$ and } |V'| = k \}.
    \end{equation*}
    Moreover, we define the set of all independent sets of $H$ as
    \begin{equation*}
        \Ind(H) = \{ V' \subset V \mid V' \text{ is an independent set in $H$}\}.
    \end{equation*}
\end{definition}

\begin{definition}[Independence Number]
    Let $H = (V,E)$ be an $r$-uniform hypergraph. Then the independence number of $H$, denoted as $\alpha(H)$, is the size of largest independent set in $H$.
\end{definition}

We denote the complete $r$-uniform hypergraph on $n$ vertices as $\KK{n}{r}$.

\subsection{Hypergraph $\turan$ Problem}

In this subsection, we will talk about the $\turan(n, s, r)$ systems.

\begin{definition}[$\turan(n,s,r)$ system]
    Let $H = (V,E)$ be an $r$-uniform hypergraph on $n$ vertices, then $H$ is called $\turan(n,s,r)$ system if $\alpha(H) = s-1.$ Moreover, define
    \begin{equation*}
        T(n,s,r) = \min\{|E(H)| \mid H \text{ is a $\turan(n,s,r)$ system over $n$ vertices}\},
    \end{equation*}
    which is the size of the smallest hypergraph on $n$ vertices with independence number $s-1.$
\end{definition}

In the combinatorics literature it is more common to study $ex(n,s,r) = \binom{n}{r} - T(n,s,r)$, which is the problem of studying the largest $r$-uniform hypergraphs that avoids the complete $r$-uniform hypergraph on $s$ vertices, $\KK{s}{r}$.\\

Let $t(s,r) = \displaystyle\lim_{n \to \infty}\frac{T(n,s,r)}{\binom{n}{r}}$.
It is known that $t(s,r)$ exists and $T(n,s,r) \leq \binom{n}{r}\cdot t(s,r)$ for $n \geq s$.\\

For the case $r =2$ (i.e., graphs), it follows from \turan's Theorem~\cite{Turan1941} that $t(s,2) = \frac{1}{s-1}$. For $s > r \geq3$, determining $t(s,r)$ is a major open problem. The problem has a very rich history and we refer the readers to \cite{Keevash_2011} for a more comprehensive survey of the problem.\\

We will require the following theorem proved by Gao, Kuang, Pikhurko and Wang (\cite{gao2026improved}).

\begin{theorem}[\cite{gao2026improved}]
\label{thm:Turan(r+1-r)-explicit-construction}
    Let $r \geq 2$ be fixed. For sufficiently large $n$, there is an explicit construction of a $\turan(n,r+1,r)$ system $H$ with:
    \begin{equation*}
        |E(H)| \leq \left( \frac{2}{1 + \lfloor \frac{r+1}{2}\rfloor} + o(1)\right)\cdot \binom{n}{r}.
    \end{equation*}
\end{theorem}

\subsection{Fourier Analysis}
We introduce basic notation and recall facts from Fourier analysis in this subsection and refer to \cite{kopparty2019fourier} for a more comprehensive treatment of the topic. Let $G$ be an abelian group. A group homomorphism $\psi: G \to \C^{\times}$ is called the character of the group $G$. For a finite abelian group $G$, all the characters themselves form a group under pointwise multiplication and this group is isomorphic to $G$.

For example, the characters of the underlying additive group of $\F_q$ are $\{\psi_a: \F_q \to \C^{\times}\}_{a \in \F_q}$ where,
$$\psi_a(x) = \omega_p^{\Tr_{\F_q/\F_p}(ax)},$$
where the characteristic of $\F_q$ equals $p$, $\omega_p$ is a primitive $p$'th root of $1$ in $\mathbb C$, and $\Tr$ is the field trace.
(In the case of $q$ being prime, the $\Tr$ map is simply identity).
Below, we recall the Fourier Transform with respect to the additive group of $\F_q$ (it can be similarly defined more generally over finite abelian groups, but we won't need it).

These characters $\psi_a$ can be viewed as vectors in $\C^{\F_q}$. We define a notion of inner product on this space,
\begin{definition}
    For two vectors $u, v \in \C^{\F_q}$, we define their inner product $\langle u, v\rangle$ as,
    $$\langle u, v\rangle = \expectation_{x \in \F_q}[u(x) \overline{v(x)}]$$
\end{definition}
The characters $\psi_a$ form an orthonormal basis for the vector space $\C^{\F_q}$.
\begin{lemma}
Let $\{ \psi_a : a \in \F_q\}$ be the characters of the additive group $\F_q$. Then,
\[
    \langle \psi_a, \psi_b \rangle =
    \begin{cases}
    1, & a = b,\\
    0, & a \neq b.
    \end{cases}
\]
\end{lemma}
So, a function $F: \F_q \to \C$ can be written as a linear combination of the characters $\psi_a$,
$$F = \sum_{a \in \F_q} \widehat F(a) \psi_a.$$
The scalars $\widehat F(a) \in \C$ are known as the Fourier coefficients of $F$. Since the characters $\psi_a$ are orthonormal, we can obtain the Fourier coefficients $\widehat F: \F_q \to \C$ from the following inner product,
$$\widehat F(a) = \langle F, \psi_a \rangle = \expectation_{x \in \F_q} [F(x) \overline{\psi_a(x)}].$$ 
Since the ``change of basis'' linear operation to obtain $\widehat F$ from $F$ is a unitary operation, the norm of $\widehat F$ is equal to the norm of $F$. In literature, this fact is stated as the Parseval identity.
\begin{lemma}[Parseval identity]
    \label{lem:parseval-identity}
    Let $F: \F_q \to \C$. Then,
    $$\langle F, F \rangle = \expectation_{x \in \F_q}[|F(x)|^2] = \sum_{a \in \F_q} |\widehat F(a)|^2$$
\end{lemma}

\section{Nonlinearity, $\magree{F}{d}$, and the Fourier bound}
\label{sec:non-linearity}

In this section, we will formally define the nonlinearity of a function and its extension to larger degrees. First we define affine functions below.

\begin{definition}[Affine Functions]
    We call $L : \F_q^n \longrightarrow \F_q^m$, an $\F_q$ affine function if there exists an $m \times n$ matrix $M \in \F_q^{m \times n}$ and $b\in \F_q^m$ such that for any $x \in \F_q^n$, $L(x) = Mx + b$. We call $M$ the associated matrix of $L$ and $b$ the associated vector of $L$.
\end{definition}
Next we define the nonlinearity of any function $F : \F_q^n \longrightarrow \F_q^m$. The notion of nonlinearity we will use denoted by $\NL$ is known as vector nonlinearity or nonlinearity of the second type in the literature (\cite{CD04},\cite{CD07},\cite{chen2001nonlinearity},\cite{liu2014nonlinearity},\cite{LMC17}, \cite{BTU26}, \cite{Na25}).

\begin{definition}[Nonlinearity]
    \label{def:non-linearity}
    Let $F : \F_q^n \longrightarrow \F_q^m$ be a function, and $L :\F_q^n\longrightarrow \F_q^m$ be a $\F_q$ affine function, then we define the agreement of $F$ with respect to $L$ as
    \begin{equation*}
        \agree(F,L) = | \{x \in \F_q^n \mid F(x) = L(x) \} | .
    \end{equation*}
    We define the nonlinearity of $F$ as
    \begin{equation*}
        \NL(F) = q^n - \max_{\substack{L : \F_q^n \longrightarrow\F_q^m\\ L \text{ is an $\F_q$ affine function} }}\agree(F,L).
    \end{equation*}
\end{definition}

Computing nonlinearity is computationally very difficult. \cite{Find14} showed that computing $\NL(F)$ is $\# P$-hard, when $F$ is given as a circuit and $m=1$. \\

For our purposes, it will be more convenient to talk about the ``linearity", namely $q^n - \NL(F)$. We give this definition, along with its generalization for larger degrees below.

\begin{definition}[Max Agreement]
    Let $F : \F_q^n \longrightarrow \F_q^m$ be a function and $Q_1,...,Q_m \in \F_q[\bar X]$ be polynomials of degree at most $d$. Then we define the agreement of $F$ with respect to $\bar Q = (Q_1,...,Q_m)$ as follows
    \begin{equation*}
        \agree(F, \bar Q) = |\{ x \in \F_q^n \mid F(x) = (Q_1(x),...,Q_m(x)) \}|.
    \end{equation*}
    We define the max agreement of $F$ with respect to degree $d$ polynomials as
    \begin{equation*}
        \magree{F}{d}= \max_{\substack{\bar Q = (Q_1,...Q_m)\\ Q_i \in \F_q[\bar X] \text{ and }\deg(Q_i) \leq d}}\agree (F, \bar Q).
    \end{equation*}
\end{definition}

\begin{remark}
    For a function $F : \F_q^n \longrightarrow \F_q^m$, $\NL(F) = q^n - \magree{F}{1}$.
\end{remark}

Next, we state a lower bound on $\magree{F}{d}$ for any function $F: \F_q^n \longrightarrow \F_q^m$ and give its proof in \cref{sec:proof-of-interpolation-bound}.

\begin{lemma}[Interpolation Bound]\label{lemma:interpolation-bound}
    For any function $F : \F_q^n \longrightarrow \F_q^m$
    \begin{equation*}
        \magree{F}{d} \geq M(q,n,d).
    \end{equation*}
    where $M(q,n,d)$ is the number of $n$-variate monomials of degree at most $d$ and individual degree at most $q-1$.
\end{lemma}

\begin{remark}
    If $q > d$, then $M(q,n,d) = \binom{n+d}{d}$, otherwise by the principle of inclusion and exclusion:
    \begin{equation*}
        M(q,n,d) = \sum_{i=0}^{\min\{n, \lfloor \frac{d}{q} \rfloor \}} (-1)^i\cdot \binom{n}{i}\cdot\binom{n+d - i \cdot q}{n}.
    \end{equation*}

For $q, d$ fixed and $n \to \infty$, 
$$ M(q,n,d) = \Theta(n^d).$$
\end{remark}

\subsection{Existence of Hard Functions}

We now state the lemma that there exists function $F :\F_q^n \longrightarrow \F_q^m$ that are extremely far from any tuple of degree $d$ functions.

The proof uses the probabilistic method (indeed, a random $F$ has this property), and can be found in \autoref{sec:proof-of-existence-of-hard-functions}.

\begin{lemma}[Existence of very hard function]\label{lemma:random-function-agreement}
    Assume $m \geq n$, then there exists a function $F : \F_q^n \longrightarrow \F_q^m$, such that 
    \begin{equation*}
        \magree{F}{d} < \min\left(\frac{m}{m-n}\cdot M(q,n,d) , \quad \cO_q(M(q,n,d) \cdot m \cdot \log\big( m \cdot M(q,n,d)\big)   \right).
    \end{equation*}
    where $M(q,n,d)$ is the number of $n$-variate monomials of degree at most $d$ and individual degree at most $q-1$
\end{lemma}

\begin{remark}
    For any $\epsilon >0$, by setting $m = \cO_{\epsilon}(n)$, \cref{lemma:random-function-agreement}, implies that there exists $F: \F_q^n \longrightarrow\F_q^m$ such that $\magree{F}{d} \leq (1+\epsilon) M(q,n,d)$ which is close to the optimal (\cref{lemma:interpolation-bound}).

    For $m = n$, and $q,d, = O(1)$, we get that there exists $F$ with
    $\magree{F}{d} \leq O( M(q,n,d) \cdot n \log( M(q,n,d)) \cdot n ) = O_{q,d}(n^{d+1}\log n).$
\end{remark}

\subsection{The Basic Fourier Bound}
\label{sec:bent-function-and-basic-fourier-bound}

We now present the basic Fourier bound on nonlinearity -- which is the best bound on nonlinearity that can be deduced from ``smallness of Fourier coefficients of projections" considerations. The proofs here are over an arbitrary finite field $\F_q$.

Let $\psi : \F_q \to \mathbb C$ be a nontrivial additive character. 

\begin{lemma}
\label{lem:fourier-bound-for-far-from-linear}
Let $F: \F_q^n \to \F_q^m$ be an arbitrary function.
For $y \in \F_q^m$, define $F_y : \F_q^n \to \mathbb C$ by:
$$ F_y(x) = \psi\big( \langle y, F(x) \rangle \big).$$
Then for all affine functions $L: \F_q^n \to \F_q^m$:
$$\agree(F, L) \leq \frac{q^n}{q^m}  + q^n \cdot \max_{(y,\alpha) \neq (0,0)} |\widehat{F}_y(\alpha)|,$$
where $\alpha \in \F_q^n$.
\end{lemma}

The proof of this lemma appears in \autoref{sec:proof-fourier-bound}.

By the Parseval identity, it is well known that some non-trivial Fourier coefficient has magnitude at least $q^{-n/2}$.

\begin{lemma}\label{lemma:large-fourier-coefficient}
    For all $y \in \F_q^m \setminus\{0\}$, there exists some $\alpha \in \F_q^n $  such that, $$|\widehat{F}_y(\alpha)| \geq q^{-n/2}.$$ 
\end{lemma} 

As a consequence of \cref{lemma:large-fourier-coefficient},  there is no function $F$ for which  the basic Fourier bound, \cref{lem:fourier-bound-for-far-from-linear}, can let us conclude that  $\linagree{F} \leq o(q^{n/2})$. This is the limitation of the Fourier bound that motivates this paper.

\begin{remark}
We now give a quick sketch of why the Fourier bound (\cref{lem:fourier-bound-for-far-from-linear}) is essentially the best possible upper bound on $\linagree{F}$ in terms of smallness of Fourier coefficients of the $F_y$. Concretely, we show that there exists an $F$ which has quite large $\linagree{F}$, while having all Fourier coefficients of all the $F_y$ essentially as small as possible.

Suppose we choose $F: \F_q^n \to \F_q^m$ as follows.
Pick a uniformly random subset $S \subseteq \F_q^n$ of size $n^2 \cdot q^{n/2}$.
Define $F|_{S} = 0$. For all $x\in \F_q^n \setminus S$, choose $F(x)$ uniformly at random from $\F_q^m$.

Then the following hold with high probability (as $n \to \infty)$:
\begin{itemize}
    \item $\linagree{F} \geq q^{n-m} + \Omega(n^2 \cdot q^{n/2})$.
    \item $\max_{(y,\alpha) \neq (0,0)} |\widehat{F}_y(\alpha)| \leq O(n^2 \cdot q^{-n/2})$.
\end{itemize}

Thus, up to $\poly(n)$ factors, the Fourier bound is the best possible bound on 
$\linagree{F}$ in terms of the smallness of Fourier coefficients of the $F_y$.
\end{remark}

To demonstrate the utility of the Fourier bound, below we give a construction of a function $F$ for which the Fourier bound yields a pretty good bound on the linearity:
$$\linagree{F} \leq q^{n-m} + \cO\left(\frac{m}{n} + 1\right) \cdot  q^{n/2}.$$

\begin{example}
There is a known construction of $F : \F_2^n \longrightarrow \F_2^m$, where $m \geq n/2$, that achieves  $$\max_{(y,\alpha)\neq (0,0)}|\widehat{F}_y(\alpha)| \leq O\left(\frac{m}{n} + 1\right) \cdot q^{-n/2}.$$
which is essentially the smallest possible.

This is achieved by the following function. 
Let $\beta_1, \ldots, \beta_n$ be a basis for $\F_{q^n}$ over $\F_q$.
View $\F_q^n$ as $\F_{q^n}$ via this basis.

Let $S \subseteq \mathbb N \times [n]$ be a set of size $m$
such that for each $(i, j) \in S$, the following holds true
\begin{enumerate}
    \item $i \geq 2$
    \item $\Char(\F_q)$ does not divide $i$
    \item $i \leq 2\cdot \lceil\frac{m}{n}\rceil +2$
\end{enumerate}

Then for $x \in \F_{q^n}$, define 
$$F(x) = \bigg( \Tr(\beta_j \cdot x^i )  \bigg)_{(i,j) \in S}$$

It is a consequence of the Weil bounds for character sums that this function $F : \F_q^n \to \F_q^m$
satisfies the above claimed bound:
$$\max_{(y,\alpha)\neq (0,0)}|\widehat{F}_y(\alpha)| \leq O\left(\frac{m}{n} + 1\right) \cdot q^{-n/2}.$$
\end{example}

\section{Hard Functions over $\F_2$ using Hypergraphs}
\label{sec:hard-functions-over-F2}

In this section, we will define hard functions over $\F_2$ using hypergraphs. To bound the number of agreements of the hard function with a low degree function, we will use Macaulay's theorem (also known as the footprint bound), although our proofs won't be written in the language of Hilbert functions. This is a classical technique and can be found in \cite{GH00},\cite{Mc27},\cite{BEELEN2018130},\cite{GM17},\cite{HP98}, etc.

\subsection{Hypergraph Framework}

We will first introduce the hypergraph framework for constructing explicit functions over $\F_2$ which are far from tuples of polynomials of degree at most $r-1$. Let $H$ be an $r$-uniform hypergraph. Using $H$, we define an $\F_2$ function, denoted by $F_{H}$ and use it to study $\magree{F_{H}}{d}$, where $d < r$. In particular, we will show that $\magree{F_{H}}{d} = |\Ind(H)|$.

\begin{definition}
    Let $H = (V,E)$ be an $r$-uniform hypergraph, where $V = [n]$. Let $m = |E(H)|$ and $[m]$ be indexed by the elements of $E$. Define function $F_{H} : \F_2^n \longrightarrow \F_2^{m}$, where for any $x \in \F_2^n$ and $e = (i_1,...,i_r) \in E(H)$, define $F_{H}(x)_{e} = \displaystyle\prod_{j \in [r]} x_{i_j}$.
\end{definition}

Our main theorem for this section shows that the hardness of $F_H$ is controlled by the number of independent sets in $H$.

\begin{theorem}\label{thm:hardness-of-hypergraph-function}
    Let $H = (V,E)$ be a $r$-uniform hypergraph where $V = [n]$. Then for any $1 \leq d < r$, 
    \begin{equation*}
        \magree{F_{H}}{d} = |\Ind(H)|.
    \end{equation*}
\end{theorem}

We now make some definitions to help us reason about the hardness of $F_H$.

\begin{definition}
    Let $H = (V,E)$ be an $r$-uniform hypergraph where $V = [n]$, then define the set of indestructible monomials with respect to $H$ as
    \begin{equation*}
        \cM_{H} = \{ \bar{X}^{\textbf{1}_{W}} \mid W \in \Ind(H )  \}.
    \end{equation*}
\end{definition}

Next we define our main object of interest, the set of agreements between $F_H$ and tuples of polynomials of degree at most $r-1$.

\begin{definition}
    Let $H = (V,E)$ be an $r$-uniform hypergraph where $V = [n]$. Moreover, let $m = |E(H)|$ and $\bar{Q} = (Q_1,...,Q_m)$ be a tuple of $n$-variate polynomials of degree less than $r$. Define
    \begin{equation*}
        S_{H,\bar{Q}} =\{ x \in \F_2^n \mid F_{H}(x) = \bar{Q}(x) \},
    \end{equation*}
    to be the set of agreement between $F_{H}$ and $\bar{Q}$ in $\F_2^n$.
\end{definition}

Observe that $S_{H,\bar{Q}}$ is the set of common solutions to the following set of equations over $\F_2$:
\begin{align}
    \forall e=(i_1, \dots, i_r) \in E(H) &: F_{H}(\bar X)_{e} = Q_e(\bar X) \nonumber \\
    \iff \forall e=(i_1, \dots, i_r) \in E(H)&: \prod_{j \in [r]} X_{i_j} = Q_e(\bar X) \label{eqn:common-solutions}
\end{align}
We give a brief overview of the main ideas for proving $\magree{F_{H}}{d} = |\Ind(H)|$ (where $d < r$), before going into the technical details.

\paragraph{Upper bound for $\magree{F_H}{d}$:} We will upper bound the size of $S_{H,\bar{Q}}$. For this, we will study the space $W = \F_2^{S_{H,\bar{Q}}}$, which is the $\F_2$-linear vector space of all functions from $S_{H, \bar Q}$ to $\F_2$. 

We will pick $|\Ind(H)|$ functions in $W$ and show that they span $W$.
This will imply that $|S_{H,\bar{Q}}| = \dim(W) \leq |\Ind(H)|$, as desired.

Which $|\Ind(H)|$ functions in $W$ do we pick? The indestructible monomials.
Why do they span $W$? We show this by an algorithmic procedure.
First notice that any function $g \in W$ is computed by {\em some} polynomial
$g'(\bar X) \in \F_2[\bar X]$. We now repeatedly ``reduce" $g'(\bar X)$ while
preserving the property that $g'(\bar X)$ computes $g$. The reduction step 
looks for a monomial in $g'(\bar X)$ divisible by $F_H(\bar X)_e$, and
replaces $F_H(\bar X)_e$ in that monomial by $Q_e(\bar X)$.
In \autoref{lemma:remove-hypergraph-junk}, we show that the above process terminates and in the end, we obtain a polynomial $\hat g(\bar X) \in \F_2[\bar X]$ which is supported on the indestructible monomials $\cM_{H}$, such that $\forall x \in S_{H,\bar{Q}}$, $g(x) = g'(x) = \hat{g}(\bar x)$ (this follows from \autoref{eqn:common-solutions}). Thus, every function $g \in W$ is spanned by the indestructible monomials $\cM_{H}$, as claimed. 

\paragraph{Lower bound for $\magree{F_H}{d}$:} We will show an explicit set of tuples $\bar{Q} = (Q_1,...,Q_m)$ such that $\agree(f_H, \bar Q) \geq \abs{\Ind(H)}$. The trivial choice $\bar Q = (0, \dots, 0)$ works for this, since for each $W \in \Ind(H)$, $\textbf{1}_W \in S_{H,\bar{Q}}$.

We formally prove below that the repeated ``replacement'' of monomials using \autoref{eqn:common-solutions} in $g(\bar X) \in \F_2[\bar X]$ indeed terminates, and in the end we get a polynomial $\hat g(\bar X) \in \F_2[\bar X]$ that is supported on the indestructible monomials $\cM_H$.

\begin{lemma}\label{lemma:remove-hypergraph-junk}
    Let $g(\bar{X}) \in \F_2[\bar X]$ be a multilinear polynomial and $H =(V,E)$ be an $r$-uniform hypergraph where $V = [n]$. Moreover, let $m = |E(H)|$ and $\bar Q = (Q_1,...,Q_{m})$ be a tuple of $n$-variate polynomials of degree less than $r$. 
    Then, there exists a multilinear polynomial $\hat g(\bar X) \in \F_2[\bar X]$ such that $\hat{g}(\bar X)$ lies in the $\F_2$ span of $\cM_{H}$ and for any $x \in S_{H,\bar{Q}}$, $g(x) = \hat g(x)$.
\end{lemma}

\begin{proof}
    We will prove the lemma using induction on the degree of $g$. For the base case, if $\deg(g) = 0$, then $g$ is a constant polynomial and thus lies in the span of $\cM_H$. Setting $\hat{g}(\bar X) = g(\bar X)$ finishes the base case.\\

    Without loss of generality, we can assume $g$ is a monomial $M$ of degree $d$, since the general case follows by linearity. For the inductive step, if $M \in \cM_{H}$ then after setting $\hat{g}(\bar X) = M$, we are done. Suppose $M \notin \cM_{H}$, then there exists $e = (i_1,..,i_r) \in E(H)$ such that $M = \left(\displaystyle\prod_{j=1}^rX_{i_j}\right)\cdot  M'$, where $M'$ is a monomial of degree $d - r$. Let $\tilde g(\bar X) = Q_e(X) \cdot M'(X) \mod \langle X_i^2 - X_i \mid i \in [n] \rangle$ be the multilinear polynomial obtained by replacing $\prod_{j=1}^rX_{i_j}$ in $M$ by $Q_{e}(X)$, and multilinearizing (which never increases degree). Then $\deg(\tilde g) \leq (r-1) + (d-r) < d$. By induction, we can find $\hat g_1(\bar X)$ such that $\hat g_1(\bar X)$ lies in the $\F_2$ span of $\cM_{H}$ and for any $x \in S_{H,\bar Q}$, $\hat g_1(x) = \tilde g(x)$.\\
    
    Since for any $x \in S_{H,\bar Q}$, $\prod_{j}x_{i_j} = Q_e(x)$, it follows by definition of $\tilde g(x)$, that for any $x \in S_{H,\bar Q}$, $\tilde g(x) = g(x)$. Hence, setting $\hat g(\bar X) = \hat g_1(\bar X)$ finishes the proof.
\end{proof}
We now prove the main result of this section, \cref{thm:hardness-of-hypergraph-function}, using \autoref{lemma:remove-hypergraph-junk} above.

\begin{proof}
    Let $m = |E(H)|$ and  $\bar{Q} = (Q_1,...,Q_{m})$ be a tuple of $n$-variate polynomials of degree at most $d$. We will first show that $\magree{F_{H}}{d} \leq |\Ind(H)|$ by proving that $|S_{H,\bar Q}| \leq |\Ind(H)|$. In order to upper bound $|S_{H,\bar{Q}}|$, we will bound the dimension of $\F_2^{S_{H,\bar{Q}}}$.\\

    Take any function $g \in \F_2^{S_{H,\bar{Q}}}$. Extend it to a multilinear polynomial $g'(\bar X)\in \F_2[\bar X]$, such that for all $x \in S_{H,\bar{Q}}$, $g(x) = g'(x)$. (This can be done because of the elementary fact that every function from $\F_2^n \to \F_2$ can be expressed as a multilinear polynomial.) 

    By \cref{lemma:remove-hypergraph-junk}, there exists a multilinear polynomial $\hat g(\bar X) \in \F_2[\bar X]$ such that for all $x \in S_{H,\bar{Q}}$, $g(x) = g'(x) = \hat g(x)$, and $\hat g(\bar X)$ lies in the $\F_2$ span of $\cM_{H}$. Since $\hat g(\bar X)$ is supported on the monomials $\cM_{H}$, we have shown that every function of $\F_2^{S_{H,\bar{Q}}}$ is spanned by monomials $\cM_{H}$. This implies $\dim(\F_2^{S_{H,\bar{Q}}}) \leq |\Ind(H)|$ and hence $|S_{H,\bar{Q}}| \leq |\Ind(H)|$. \\

    To show $\magree{F_{H}}{d} \geq |\Ind(H)|$, we need to show there exists $\bar{Q} = (Q_1,...,Q_{m})$ a tuple of $n$-variate polynomials of degree at most $d$ such that $\agree(F_H, \bar Q) \geq |\Ind(H)|$. We will set $\hat Q$ to be the zero map, then we claim the following 
    \begin{claim}
        Let $\hat Q$ be the zero map, then $\agree(F_H, \hat Q) \geq |\Ind(H)| $.
    \end{claim}
    \begin{proof}[Proof of Claim]
        To prove the claim, it is enough to show that $\{ \indicator_{W} \in \F_2^n \mid W \in \Ind(H) \} \subseteq S_{H,\hat Q}$ (since $\agree(f_h, \hat Q) = |S_{H,\hat Q}|$). Recall that every $x \in S_{H,\hat Q}$ satisfies the following set of equations,
        $$\forall e=(i_1, \dots, i_r) \in E(H): \prod_{j \in [r]} x_{i_j} = \hat Q_e(x) = 0$$
        Let $W \in \Ind(H)$. Then, for each edge $e=(i_1, \dots, i_r) \in E(H)$, there exists some $j \in [r]$ such that $i_j$ is not in $W$. Thus, for every edge $e=(i_1, \dots, i_r)$, we have $\prod_{j \in [r]} x_{i_j} = 0$ for $x = \indicator_W$. Hence, $\indicator_{W} \in S_{H,\hat Q}$ and this proves the claim. 
    \end{proof}
    Thus, we have proved that $\magree{F_{H}}{d} = |\Ind(H)|$.

\end{proof}

\autoref{thm:hardness-of-hypergraph-function} gives us a \emph{Direct Product Lemma} in our hypergraph framework, which will play a crucial role in constructing explicit hard functions in \cref{sec:instantiate-hypergraphs}.

\begin{lemma}[Direct Product Lemma]\label{lemma:direct-product-F2}
    Let $H = (V_1 \sqcup V_2 \sqcup \dots \sqcup V_t,E_1 \sqcup E_2 \sqcup \dots \sqcup E_t)$ be a disjoint union of $r$-uniform hypergraphs $H_1=(V_1, E_1), H_2 = (V_2, E_2), \dots, H_t=(V_t, E_t)$. Then,
    \begin{equation*}
        \magree{F_H}{d} = \prod_{i=1}^t \magree{F_{H_i}}{d}
    \end{equation*}
    where $F_{H_i}:\F_2^{|V_i|}\longrightarrow \F_2^{|E_i|}$ and $d < r.$
\end{lemma}

\begin{proof}
    Since $|\Ind(H)| = |\Ind(H_1)| \cdot |\Ind(H_2)| \cdots |\Ind(H_t)|$, the lemma follows by \cref{thm:hardness-of-hypergraph-function}.
\end{proof}

\subsection{Instantiating the Hypergraphs}
\label{sec:instantiate-hypergraphs}

To construct explicit hard functions $F : \F_2^n \longrightarrow \F_2^m$, $\cref{thm:hardness-of-hypergraph-function}$ motivates the following question.

\begin{question}\label{ques:min-ind-sets}
    What $r$-uniform hypergraphs on $n$ vertices and at most $m$ edges have the minimum number of independent sets?
\end{question}

For the case of graphs, the question was solved for the case $m \leq n$, by Cutler and Radcliffe (\cite{Cutler_Radcliffe_2011}). Later, Sah, Sawhney, Stoner and Zhao (\cite{SAH2019172}) in their breakthrough work, solved the problem when the degree distribution of the vertices is given. In particular, they proved the following.

\begin{theorem}[\cite{SAH2019172}]
    \label{thm:optimal-number-of-independent-sets-graphs}
    Let $G = (V,E)$ be a graph and $d_v$ be the degree of vertex $v$ in G. Then
    \begin{equation*}
        |\Ind(G)| \geq \prod_{v \in V(G)}(d_v +2)^{1/(d_v +1)}
    \end{equation*}
    and the equality holds if and only if $G$ is a disjoint union of cliques.
\end{theorem}

Motivated by the above theorem, we instantiate the hypergraphs with the disjoint union of equal-sized cliques and get the following result.

\begin{lemma}[Disjoint Union of Cliques]\label{lemma:disjoint-cliques-function}
    Let $k |n$ and $r \leq k$. Let $H_{clique,k,r}$ be the disjoint union of $n/k$ $\KK{k}{r}$ then $|E(H_{clique,k,r})| = \binom{k}{r}\cdot \frac{n}{k} $ and for any $d < r$,
    \begin{equation*}
        \magree{F_{H_{clique}{k,r}}}{d} = \left( \sum_{i=0}^{r-1}\binom{k}{i} \right)^{n/k},
    \end{equation*}
    where $F_{H_{clique,k,r}}: \F_2^{n} \to \F_2^{|E(H_{clique,k,r})|}$.
\end{lemma}

\begin{proof}
        Note that $|E(\KK{k}{r})| = \binom{k}{r}$ and thus $|E(H_{clique,k,r})| = \binom{k}{r}\cdot \frac{n}{k}$. Furthermore observe that 
        \begin{equation*}
            |\Ind(\KK{k}{r})| = \sum_{i=0}^{r-1} \binom{k}{i},
        \end{equation*}
         and thus by \cref{lemma:direct-product-F2}, for any $d < r$,
    \begin{equation*}
        \magree{F_{H_{clique}{k,r}}}{d} = \left( \sum_{i=0}^{r-1}\binom{k}{i} \right)^{n/k}
    \end{equation*}
    which finishes the proof.
\end{proof}

We can construct explicit non-linear functions over $\F_2$ by instantiating \autoref{lemma:disjoint-cliques-function} with the appropriate set of parameters.

\begin{theorem}[Explicit Non-Linear Functions beyond the Fourier Bound]
    \label{thm:explicit-non-linear-functions-beyond-fourier-bound}
    Let $k|n$. Suppose $G=(V,E)$ is a graph such that $|V|=n$ and $G$ is a disjoint union of $n/k$ cliques $K_k$. Then,
    $$\magree{F_G}{1} = (1+k)^{n/k},$$
    where $F_G: \F_2^n \to \F_2^{n(k-1)/2}$.
\end{theorem}
\begin{proof}
    Instantiate \autoref{lemma:disjoint-cliques-function} for $r=2$ and $d=1$.
\end{proof}
\begin{remark}
    In particular, when $k=n$ (i.e., $G = K_n$),
    $$\magree{F_{K_n}}{1} = n+1,$$
    where $F_{K_n}: \F_2^n \to \F_2^{n(n-1)/2}$. By \cref{lemma:interpolation-bound}, any function $g: \F_2^n \to \F_2^m$ has agreement at least $(n+1)$ with affine functions. Thus, $F_{K_n}$ achieves this optimal agreement.
\end{remark}
\begin{remark}
    Setting $k = 6$, we obtain function $F_G: \F_2^n \to \F_2^{5n/2}$ with $\magree{F_G}{1} = 7^{n/6} \approx 1.383^n$. This beats the agreement of $2^{n/2} \approx 1.414^n$ which one gets by the Fourier method (refer \cref{sec:bent-function-and-basic-fourier-bound}).

    In fact, for every constant $c$, we can choose a large enough constant $k$ such that $\magree{F_G}{1} \leq 2^{n/c}$, where $F_G: \F_2^n \to \F_2^{\cO(kn)}$. A limitation of our (hyper)graph framework is that we cannot get a sub-exponential agreement for non-linear functions in the regime $m=O(n)$ (due to the optimality of our construction for graphs from \autoref{thm:optimal-number-of-independent-sets-graphs}).
\end{remark}

We can also construct explicit functions over $\F_2$ which are far from tuples of degree $d$ polynomials.

\begin{theorem}
    \label{thm:explicit-non-linear-functions-beyond-fourier-bound-arbitrary-d}
    Let $k|n$. Suppose $H=(V,E)$ is a $d+1$-uniform hypergraph such that $|V|=n$ and $H$ is a disjoint union of $n/k$ cliques $\KK{k}{d+1}$. Then,
    $$\magree{F_H}{d} = \left(\sum_{i=0}^d \binom{k}{i}\right)^{n/k},$$
    where $F_H: \F_2^n \to \F_2^{\frac{n}{k}\cdot \binom{k}{d+1}}$.
\end{theorem}
\begin{proof}
    Instantiate \autoref{lemma:disjoint-cliques-function} for $r=d+1$.
\end{proof}

\begin{corollary}[Functions with polynomial agreement with low degree functions]
    Let $d$ be a large constant. Suppose $\KK{n}{d+1}$ be a $(d+1)$-uniform clique. Then,
    $$\magree{F_{\KK{n}{d+1}}}{d} = \cO(n^{d}),$$
    where $F_{\KK{n}{d+1}}: \F_2^{n} \to \F_2^{\cO(n^{d+1})}$. 
\end{corollary}
\begin{proof}
    Instantiate \autoref{thm:explicit-non-linear-functions-beyond-fourier-bound-arbitrary-d} with $k=n$.
\end{proof}
\begin{remark}
    By \cref{lemma:interpolation-bound}, any function $g: \F_2^n \to \F_2^m$ has agreement at least $\Omega(n^{d})$ with tuples of degree $d$ polynomials. Thus, $F_{\KK{n}{d+1}}$ achieves this optimal agreement.
\end{remark}

We will next show that for a large enough constant $r$, the disjoint union of cliques is not the best extremal solution to \autoref{ques:min-ind-sets}. In particular, consider the following.

\begin{lemma}\label{lemma:disjoint-turan(r+1-r)-ind-count}
   Let $r\geq 2$ be a fixed number. For a sufficiently large $k,n$ where $k|n$, there exists an explicit $r$-uniform hypergraph $H_{\turan,k,r}$ over $n$ vertices where
   \begin{equation*}
       |E(H_{\turan,k,r})|= \lambda\cdot \binom{k}{r}\cdot \frac{n}{k}
   \end{equation*}
   and 
   \begin{equation*}
       |\Ind(H_{\turan,k,r})| \leq \left( \left(1 - \lambda+ o(1)\right)\cdot \binom{k}{r} \right)^{n/k}.
   \end{equation*}
   Moreover, 
   \begin{equation*}
       \lambda \leq \left( \frac{2}{1+ \lfloor\frac{r+1}{2}\rfloor} + o(1) \right).
   \end{equation*}
\end{lemma}

\begin{proof}
    Let $\hat H$ be the $r$-uniform graph over $k$ vertices from \cref{thm:Turan(r+1-r)-explicit-construction}. Let $\lambda = \frac{|\hat H|}{\binom{k}{r}}$, then since $\hat H$ is $\turan(k,r+1,r)$ system, the following holds true 
    \begin{equation*}
        |\Ind(\hat H)| = \sum_{i=0}^{r-1}\binom{k}{i} + (1- \lambda)\cdot \binom{k}{r}
    \end{equation*}
    \cref{thm:Turan(r+1-r)-explicit-construction} also implies that
    \begin{equation*}
        \lambda \leq \left( \frac{2}{1+ \lfloor\frac{r+1}{2}\rfloor} + o(1) \right).
    \end{equation*}
    Let $H$ be the disjoint union of $n/k$ copies of $\hat H$, then 
    \begin{align*}
        |\Ind(H)| &= \left( \sum_{i=0}^{r-1} \binom{k}{i} + (1- \lambda)\binom{k}{r} \right)^{n/k}\\
        &\leq \left( \left(1 - \lambda- o(1)\right)\cdot \binom{k}{r} \right)^{n/k}.
    \end{align*}
    The explicitness guarantees of the lemma follows from \cref{thm:Turan(r+1-r)-explicit-construction}.
    
\end{proof}

As a corollary of the above lemma, we can prove the following.

\begin{corollary}[Disjoint Union of $\turan(n,r+1,r)$ systems]\label{corollary:Disjoint-Union-Turan-Systems}
    Let $d \geq 1$ be fixed. For large enough $k,n$ such that $k \mid n$, then  $F_{H_{\turan,k,d+1}} : \F_2^n \longrightarrow \F_2^m$, where $m= \lambda \cdot \binom{k}{d+1}\cdot \frac{n}{k}$
    for 
    \begin{equation*}
        \lambda \leq \left( \frac{2}{1+ \lfloor\frac{d+2}{2}\rfloor} + o(1) \right)
    \end{equation*}
    and 
    \begin{equation*}
        \magree{F_{H_{\turan}{k,d+1}}}{d} \leq \left( \left(1 - \lambda- o(1)\right)\cdot \binom{k}{d+1} \right)^{n/k}.
    \end{equation*}
\end{corollary}

\begin{proof}
    Instantiate \cref{lemma:disjoint-turan(r+1-r)-ind-count} with $r = d+1$.
\end{proof}

We will now show that we can beat the disjoint union of clique construction for large enough constant $r$. In particular, let $k$ be a large enough number and 
\begin{equation*}
    \lambda \leq \left(  \frac{2}{1 + \lfloor\frac{r+1}{2} \rfloor} + o(1) \right)
\end{equation*}
be the constant from \cref{lemma:disjoint-turan(r+1-r)-ind-count}. If we set $a  = \lceil k \cdot \lambda^{1/(r-1)}\rceil + r-1$, and $n$ be a large enough number such that $k \mid n$ and $a \mid n$, then by \cref{lemma:disjoint-turan(r+1-r)-ind-count} and \cref{lemma:disjoint-cliques-function},
\begin{align*}
    \frac{|E(H_{\turan,k,r})|}{|E(H_{clique,a,r})|} &\leq \frac{\lambda \cdot \binom{k}{r} \cdot \frac{n}{k} }{ \binom{a}{r}\cdot \frac{n}{a}}= \frac{\lambda \cdot \binom{k-1}{r-1}}{\binom{a-1}{r-1}}\\
    &\leq \lambda \cdot \prod_{i=1}^{r-1} \frac{k-i}{a-i}\\
    &< \lambda \cdot \frac{k^{r-1}}{\lambda\cdot k^{r-1} } =1. 
\end{align*}
In the following lemma we will compare $|\Ind(H_{clique,a,r)}|$ and $|\Ind(H_{turan,k,r})|$ and show that $|\Ind(H_{\turan,k,r})| << |\Ind(H_{clique,a,r})|$ for fixed $r$.

\begin{lemma}
    Let $r\geq2$ be a fixed number. Let $k$ be a large enough number that is growing and $\lambda \leq \left(  \frac{2}{1 + \lfloor\frac{r+1}{2} \rfloor} + o(1) \right)$ is the constant from \cref{lemma:disjoint-turan(r+1-r)-ind-count}. Let $a = \lceil k \cdot \lambda^{1/(r-1)}\rceil + (r-1)$ and $n$ be a large enough number such that $k\mid n$ and $a \mid n$. Then
    \begin{equation*}
        |\Ind(H_{\turan,k,r})| \leq |\Ind(H_{clique,a,r})|^{\frac{r}{r-1}\lambda^{1/(r-1)}\left( 1 + o(1) \right)}
    \end{equation*}
\end{lemma}

\begin{proof}
    For ease of notation let $H_1 = H_{\turan,k,r}$ and $H_2 = H_{clique,a,r}$, then by \cref{lemma:disjoint-turan(r+1-r)-ind-count} and \cref{lemma:disjoint-cliques-function}, we know that
    \begin{align*}
        \frac{\log(|\Ind(H_1)|)}{\log(|\Ind(H_2)|)} &\leq \frac{a}{k}\cdot \frac{\log\left((1-\lambda+o(1))\cdot \binom{k}{r}  \right)}{\log\left( \displaystyle\sum_{i=0}^{r-1} \binom{a}{i} \right)}\\
        &\leq \frac{a}{k}\cdot \frac{\log\binom{k}{r}}{\log\binom{a}{r-1}}\\
        &\leq \frac{r}{r-1}\cdot \lambda^{1/(r-1)} \cdot (1 + o(1))
    \end{align*}
    The above inequality implies
    \begin{equation*}
        |\Ind(H_1)| \leq |\Ind(H_2)|^{\frac{r}{r-1}\cdot \lambda^{1/(r-1)}\cdot (1 + o(1))}.
    \end{equation*}
    and hence we are done with the proof.
\end{proof}

It is not clear to us what extremal family of hypergraphs will answer \autoref{ques:min-ind-sets}.

We will now give an explicit construction of a family of sparse 3-uniform hypergraph $H^{*}$ over $8$ vertices with 115 independent sets.

\begin{definition}\label{def:8-vertex-H*}
    Let $H^{*} = (V^{*},E^{*})$ be a 3-uniform hypergraph where $V^{*} = \F_3^2 \setminus \{(0,0)\}$. $E^*$ consists of affine lines not passing through the origin. In particular
    \begin{equation*}
          E^* = \{ \{x,y,z\} \in \left(\F_3^2 \setminus \{(0,0)\}\right )^3 \mid  \text{ $x,y,z$ are distinct and }x +y+z =0\}
    \end{equation*}
    where $|E^*| = 8$.
\end{definition}

Now we will count the number of independent sets in $H^*$, in the following proposition.

\begin{proposition}\label{prop:H*-independent-set-count}
    $|\Ind(H^*)| = 115.$
\end{proposition}

\begin{proof}
    Notice that any subset of size at most $2$, is an independent set. Thus $|\Ind(H^*,0)| =1$, $|\Ind(H^*,1)| = 8$ and $|\Ind(H^*,2)| = \binom{8}{2} = 28$. The independent sets of size $3$ are the non edges and thus $|\Ind(H^*,3)| = \binom{8}{3} - 8 = 48$.\\
    
    To find $|\Ind(H^*,4)|$, we count the number of sets of size $4$ that does not contain an affine line. Note that any set of size $4$ will contain at most $1$ affine line. Thus, the number of sets of size $4$ with an affine line is $8 \cdot 5 = 40$ which implies $|\Ind(H^*,4)| = \binom{8}{4} - 40 = 30$.\\

    Now we will show that any set of size $5$ will contain an affine line. For the sake of contradiction, if there exists a set of $5$ points with no affine line, then any unordered pair of points will give a distinct line which gives us $\binom{5}{2}=10$ distinct lines.. Since there are only 4 possible directions for lines, then by the pigeonhole principle, there exists 3 parallel lines such that $S$ contains two points of each line. However, this gives a contradiction as then $|S|$ needs to be at least 6. Hence $|\Ind(H^*,t)|=0$ for $t \geq 5$, and thus
    \begin{equation*}
        |\Ind(H^*)| = \sum_{i=0}^8 |\Ind(H^*,i)| = 115
    \end{equation*}
    which finishes the proof.
\end{proof}

From this sparse uniform hypergraph with very few independent sets, we can construct hard explicit functions $F: \F_2^n \longrightarrow \F_2^n$ using the direct product lemma (\cref{lemma:direct-product-F2}). In particular, we prove the following lemma

\begin{lemma}\label{lemma:sparse-hypergraph-many-ind-sets}
    Suppose $8 |n$, then there exists an explicit function $F: \F_2^n \longrightarrow \F_2^n$ such that $\magree{F}{2} = 115^{n/8}$
\end{lemma}

\begin{proof}
    Let $H$ be the 3-uniform hypergraph consisting of $n/8$ disjoint copies of $H^{*}$. Then by \cref{lemma:direct-product-F2}, it follows that
    \begin{equation*}
        \magree{f_H}{2} = (\magree{F_{H^*}}{2})^{n/8}.
    \end{equation*}
    Set $F = F_{H}$, then by \cref{thm:hardness-of-hypergraph-function} and \cref{prop:H*-independent-set-count}, $\magree{F}{2} = 115^{n/8}$, which finishes the proof.
\end{proof}

\begin{remark}
    Note that if we take the disjoint union of $n/4$ copies $\KK{4}{3}$, then the corresponding hypergraph will have $n$ edges and $11^{n/4}$ independent sets. However if we take $n/8$ copies of $H^*$, then the corresponding hypergraph will have $n$ edges and $115^{n/8}$ ($< 11^{n/4}$) independent sets. This shows that disjoint union of hypercliques are not the solution to \autoref{ques:min-ind-sets} in the sparse regime.
\end{remark}

We discuss the limitations of the hypergraph framework in \autoref{sec:limitations-of-hypergraph-framework}.

\section{Hard Functions over $\F_q$ using Hypergraphs}
\label{sec:hard-functions-over-Fq}

As in \cref{sec:hard-functions-over-F2}, we will use the hypergraph framework to build hard functions over $\F_q$. But unlike the $\F_2$ case, the agreements will not be characterized neatly by the number of independent sets in the hypergraphs. We will define a new quantity associated to the hypergraph $H$, called $N_{H}(q)$, that will characterize the agreements over $\F_q$. 

\subsection{Hypergraph Framework}

Let $H$ be an $r$-uniform hypergraph. We define the $\F_q$ function corresponding to $H$, denoted by $F_{H,q}$, as follows. We first define a set of tuples $M_{H,q}$ that will be helpful in defining each output component of $F_{H,q}$.

\begin{definition}[Components of Hypergraph Function]
    Let $H = (V,E)$ be an $r$-uniform hypergraph, where $V = [n]$. Let $q$ be any prime power and $e \in E( H)$, where $e = \{i_1,...,i_r\}$. Define $M_{e,q}$ as follows 

    \begin{equation*}
        M_{e,q} = \left\{ T \in \{0,1,...,q-1\}^n \mid T _{|_{[n] \setminus \{i_1,...,i_r\}} } = 0, \text{ and } \sum_{j\in [r]} T_{i_j} = r  \right\}
    \end{equation*}
    Further, define the set $M_{H,q} = \bigcup_{e \in E} M_{e,q}$.
\end{definition}
\begin{remark}
    In the above definition, two distinct edges $e_1$ and $e_2$ might contain common elements in their corresponding sets $M_{e_1, q}$ and $M_{e_2, q}$: $M_{e_1, q} \cap M_{e_2, q} \neq \phi$. This redundancy is taken care of when we define $M_{H,q} = \bigcup_{e \in E} M_{e,q}$.
\end{remark}

\begin{definition}[Hypergraph Function over $\F_q$]
    Let $H = (V,E)$ be an $r$-uniform hypergraph, where $V = [n]$. Let $q$ be any prime power, and $m_{H,q} = |M_{H,q}|$, where we index $[m_{H,q}]$ by the elements of $M_{H,q}$. We define the function $F_{H,q} : \F_q^{n} \longrightarrow \F_q^{m_{H,q}}$ where for any $x \in \F_q^n$ and $T \in M_{H,q}$, 
    \begin{equation*}
        F_{H,q}(x)_T = \prod_{i=1}^n x_i^{T_i}.
    \end{equation*}
\end{definition}

In \cref{sec:hard-functions-over-F2}, the hardness of the hypergraph function was characterized by $\Ind(H)$. Over $\F_q$, this role will be played by $N_H(q)$, which we define below.

\begin{definition}
    Let $H=(V,E)$ be an $r$-uniform hypergraph, where $V =[n]$. Let $q$ be a prime power, then for any $e\in E$, where $e = \{i_1,...,i_r\}$ define
    \begin{equation*}
        N_H(q,e) = \left\{ T \in \{ 0,1,...,q-1 \}^n \mid \sum_{j \in [r]} T_{i_j} < r \right\}.
    \end{equation*}
    Also  define the set $N_H(q) = \bigcap_{e \in E} N_H(q,e)$.
\end{definition}
\begin{remark}
    Observe that when $q=2$, $N_H(2) = \Ind(H)$. Hence, $N_H(q)$ generalizes the notion of independent sets.
\end{remark}

Next, we define our main object of interest, the set of agreements between $F_{H, q}$ and tuples of polynomials of degree at most $r-1$.

\begin{definition}
    Let $H = (V,E)$ be an $r$-uniform hypergraph where $V= [n]$. Let $\bar{Q} = (Q_1,...,Q_{m_{H,q}})$ be $n$-variate polynomials of degree less than $r$, where $m_{H,q} = |M_{H,q}|$. Then define
    \begin{equation*}
        S_{H,\bar Q, q} = \{ x \in \F_q^n \mid F_{H,q}(x) = \bar{Q}(x) \},
    \end{equation*}
    the set of agreement between $F_{H,q}$ and $\bar{Q}$ in $\F_q^n.$
\end{definition}

Observe that $S_{H,\bar{Q},q}$ is the set of common solutions to the following set of equations over $\F_q$:
\begin{align}
    \forall T \in M_{H,q} &: F_{H,q}(\bar X)_{T} = Q_T(\bar X) \nonumber \\
    \iff \forall T \in M_{H,q}&: \prod_{i \in [n]} X_{i}^{T_i} = Q_T(\bar X) \label{eqn:common-solutions-Fq}
\end{align}

As in \cref{sec:hard-functions-over-F2}, we will bound the size of $S_{H,\bar{Q},q}$ to analyze $\magree{F_{H,q}}{d}$ (where $d < r$). For this, we define the set of indestructible monomials with respect to the hypergraph $H$. 

\begin{definition}[Indestructible monomials with respect to $H$]
    The set of indestructible monomials with respect to $H$ is defined as
    \begin{equation*}
    \cM_{H,q} = \{ X_1^{T_1}\cdot \cdot \cdot X_n^{T_n} \mid T = (T_1,...,T_n) \in N_H(q) \}.
    \end{equation*}
    In particular, note that $|\cM_{H,q}| = |N_{H}(q)|.$
\end{definition}

We will now bound $\magree{F_{H,q}}{d}$ for $d < r$. Similar to \autoref{lemma:remove-hypergraph-junk}, we first formally prove below that the repeated ``replacement'' of monomials using \autoref{eqn:common-solutions-Fq} in $g(\bar X) \in \F_q[\bar X]$ indeed terminates, and in the end we get a polynomial $\hat g(\bar X) \in \F_q[\bar X]$ that is supported on the indestructible monomials $\cM_{H, q}$.

\begin{lemma}\label{lemma:remove-hypergraph-junk-Fq}
    Let $g(\bar{X}) \in \F_q[\bar X]$ be a polynomial of individual degree less than $q$ and $H =(V,E)$ be an $r$-uniform hypergraph where $V = [n]$. Moreover, let $m_{H,q} = |M_{H,q}|$ and $\bar Q = (Q_1,...,Q_{m_{H,q}})$ be a tuple of $n$-variate polynomials of degree less than $r$. 
    Then, there exists a polynomial $\hat g(\bar X) \in \F_q[\bar X]$ of individual degree less than $q$ such that, $\hat g(\bar X)$ lies in the $\F_q$ span of $\cM_{H,q}$ and for any $x \in S_{H,\bar{Q},q}$, $g(x) = \hat g(x)$.

\end{lemma}

\begin{proof}
    We prove the lemma by induction on the degree of $g$. For the base case, if $\deg(g) = 0$, then $g$ is a constant polynomial which implies $g(\bar X)$ lies in the $\F_q$ span of $\cM_{H,q}$. Setting $\hat g(\bar X) = g(\bar X)$ finishes the proof.\\

    For the inductive case, assume without loss of generality, that $g$ is a monomial $M$ of degree $d$. The general case will follow by linearity. If $M \in \cM_{H,q}$, then setting $\hat g(\bar X) = g(\bar X)$ will finish the proof. Suppose $M \notin \cM_{H,q}$, then there exists $T \in M_{H,q}$ such that $M = \bar X^T \cdot M'$ where $M'$ is a monomial of degree $d-r$. Let $\tilde g(\bar X) = Q_T(\bar X) \cdot M' \mod \langle X_i^q- X_i \mid i \in [n] \rangle$ be a polynomial of individual degree at most $q-1$, then $\deg(\tilde g) \leq r-1 + (d-r) <d$. By induction, there exists $\hat g_1(\bar X)$ such that $\hat g_1(\bar X)$ lies in the $\F_q$ span of $\cM_{H,q}$ and for any $x \in S_{H,\bar Q,q}$, $\hat g_1(x) = \tilde g(x)$.\\

    Since for any $x \in S_{H,Q,q},$ $x^T = Q_T(x)$, it follows by the definition of $\tilde g$, that $\tilde g(x) = g(x)$. Hence setting $\hat g(\bar X) = \hat g_1(\bar X) $ finishes the proof.
\end{proof}

We now prove the main result that $\magree{F_{H, q}}{d} \leq |N_H(q)|$ using the above \autoref{lemma:remove-hypergraph-junk-Fq}.

\begin{theorem}\label{thm:hardness-of-hypergraph-function-Fq}
    Let $H = (V,E)$ be a $r$-uniform hypergraph where $V = [n]$ and let $1 \leq d < r$. Then, 
    \begin{equation*}
        \magree{F_{H, q}}{d} \leq |N_H(q)|.
    \end{equation*}
    Moreover, $\magree{F_{H, q}}{r-1} = |N_{H}(q)|$.
\end{theorem}

\begin{proof}
    Let $m_{H,q} = |M_{H,q}|$ and  $\bar{Q} = (Q_1,...,Q_{m_{H,q}})$ be a tuple of $n$-variate polynomials of degree at most $d$, where $m_{H,q} = |M_{H,q}|$. We will show $|S_{H,\bar Q,q}| \leq |N_{H}(q)|$ to prove $\magree{F_{H, q}}{d} \leq |N_H(q)|$. In order to upper bound $|S_{H,\bar{Q},q}|$, we will bound the dimension of $\F_q^{S_{H,\bar{Q},q}}$.\\

    Take any function $g \in \F_q^{S_{H,\bar{Q},q}}$, then we can extend $g$ to a polynomial $g'(\bar X)\in \F_q[\bar X]$ of individual degree less than $q$, such that for any $x \in S_{H,\bar{Q},q}$, $g(x) = g'(x).$ By \cref{lemma:remove-hypergraph-junk-Fq}, there exists a  polynomial $\hat g(\bar X) \in \F_q[\bar X]$ of individual degree less than $q$, such that for any $x \in S_{H,\bar{Q},q}$, $g(x) = g'(x) = \hat g(x)$, and $\hat g^(\bar X) $ lies in the $\F_q$ span $\cM_{H,q}$. Hence we have shown that any function in $\F_q^{S_{H,\bar{Q},q}}$ can be  spanned by elements of $\cM_{H,q}$, which implies $\dim(\F_q^{S_{H,\bar{Q},q}}) \leq |N_H(q)|$.\\

    In order to prove that $\magree{F_{H, q}}{r-1} = |N_{H}(q)|$, we will construct $\hat{Q} = (\hat Q_T)_{T \in M_{H,q}}$, where each $\hat Q_T$ is a $n-$variate polynomial of degree at most $r-1$. Let $\F_q = \{\alpha_0,...,\alpha_{q-1}\}$ be an enumeration of the field elements. For any $T \in M_{H,q}$, define the polynomial
    \begin{equation*}
        P_T(\bar{X}) = \prod_{i=1}^n \prod_{j=0}^{T_i-1}(X_i - \alpha_j)
    \end{equation*}
    and $\hat Q_T(\bar{X}) = \bar{X}^{T} - P_T(\bar{X}).$ By design of $P_T(\bar{X})$, we have $\deg(\hat Q_T(\bar{X})) \leq r-1.$\\
    
    For any $\beta \in N_{H}(q)$, define $a(\beta) = (\alpha_{\beta_1},...,\alpha_{\beta_n}) \in \F_q^n$. To finish the proof, it suffices to show that for any $T \in M_{H,q}$ and $\beta \in N_{H}(q)$, $P_T(a(\beta))=0$ since it would imply that $\agree(F_{H,q},\hat Q) \geq |N_{H}(q)|.$ Take any $\beta \in N_{H}(q)$, $T \in M_{H,q}$ and take edge $e \in E$ such that $T \in M_{e,q}$. Let $e = (i_1,...i_r)$, then there exists $j \in [r]$ such that $\beta_{i_j} < T_{i_j}$ by definition of $N_{H}(q).$ Hence $a(\beta)_{i_j} \in \{\alpha_0,...,\alpha_{T_{i_j}-1}\}$, and thus $P_T(a(\beta)) = 0.$
\end{proof}

\begin{remark}
    The polynomial $P_T(\bar X)$ in the proof of \cref{thm:hardness-of-hypergraph-function-Fq} also appears in \cite{BEELEN2018130},\cite{GM17},\cite{HP98}.
\end{remark}

As in \cref{sec:hard-functions-over-F2}, \autoref{thm:hardness-of-hypergraph-function-Fq} gives us a \emph{Direct Product Lemma} in our hypergraph framework, which will play a crucial role in constructing explicit hard functions in \cref{sec:instantiate-hypergraphs-Fq}.

\begin{lemma}[Direct Product Lemma]\label{lemma:direct-product-Fq}
    Let $H = (V_1 \sqcup V_2 \sqcup \dots \sqcup V_t,E_1 \sqcup E_2 \sqcup \dots \sqcup E_t)$ be a disjoint union of $r$-uniform hypergraphs $H_1=(V_1, E_1), H_2 = (V_2, E_2), \dots, H_t=(V_t, E_t)$. Then,
    \begin{equation*}
        \magree{F_{H, q}}{d} \leq \prod_{i=1}^t |N_{H_i}(q)|,
    \end{equation*}
    where $F_{H_i,q}:\F_q^{|V_i|}\longrightarrow \F_q^{|M_{H_i,q}|}$ and $d < r.$
    In particular, when $d=r-1$, we have,
    \begin{equation*}
        \magree{F_{H, q}}{r-1} = \prod_{i=1}^t \magree{F_{H_i, q}}{r-1}.
    \end{equation*}
\end{lemma}

\begin{proof}
    By definition $N_{H}(q) = N_{H_1}(q)\times \dots \times N_{H_t}(q)$, which implies
    \begin{equation*}
        |N_{H}(q)| = |N_{H_1}(q)| \cdots |N_{H_t}(q)|
    \end{equation*}
    and thus the lemma holds true by \cref{thm:hardness-of-hypergraph-function-Fq}.
\end{proof}

\subsection{Instantiating the Hypergraph}

\label{sec:instantiate-hypergraphs-Fq}

To construct explicit hard functions $F : \F_q^n \longrightarrow \F_q^m$, \cref{thm:hardness-of-hypergraph-function-Fq} , motivates the following question.

\begin{question}\label{ques:min-NH(q)}
    For a fixed integer $q\geq 2$, what $r$-uniform hypergraph $H$ on $n$ vertices and $|M_{H,q}| \leq m$, has the smallest value of $|N_{H}(q)|$?
\end{question}

If we restrict the question to graphs without isolated vertices, then this question is the same \autoref{ques:min-ind-sets}. Just like the case with hypergraph functions over $\F_2$, we will instantiate the hypergraph with disjoint union of equal sized cliques to get explicit functions that go beyond the Fourier Bound

\begin{theorem}[Explicit Functions beyond the fourier bound]\label{theorem:disjoint-clique-NH-count-Fq}
    Let $k|n$ and $r \leq k$. Let $H_{clique}(k,r)$ be the disjoint union of $n/k$ $\KK{k}{r}$, then
    \begin{equation*}
        |M_{H_{clique}(k,r),q}| = \frac{n}{k}\cdot \left( \sum_{i=0}^{\min\{k,\lfloor \frac{r}{q} \rfloor\}} (-1)^{i}\cdot  \binom{k}{i} \cdot  \binom{r + k-1 - iq}{k-1} \right)
    \end{equation*}
     and for any $d <r$,
    \begin{equation*}
        \magree{F_{H_{clique}(k,r),q}}{d} \leq \left( \sum_{i=0}^{\min\{k, \lfloor \frac{r-1}{q} \rfloor\}} (-1)^i \cdot  \binom{k}{i} \cdot \binom{k +r-1 -iq}{k} \right)^{n/k}
    \end{equation*}
    where $F_{H_{clique}(k,r),q} : \F_q^n \longrightarrow \F_q^{|M_{H_{clique}(k,r),q}|}.$
    Moreover, we have
    \begin{equation*}
        \magree{F_{H_{clique}(k,r),q}}{r-1} = \left( \sum_{i=0}^{\min\{k, \lfloor \frac{r-1}{q} \rfloor\}} (-1)^i \cdot  \binom{k}{i} \cdot \binom{k +r-1 -iq}{k} \right)^{n/k}.
    \end{equation*}
\end{theorem}

\begin{proof}
    Let $\hat H = \KK{k}{r}$, then $|M_{\hat H,q}|$ is the number of integral solutions of the form $X_1 + \dots  +X_k =r$, where $0 \leq X_i \leq q-1$ for $i \in [k]$. By the inclusion and exclusion principle, we get
    \begin{equation*}
        |M_{\hat H,q}| = \sum_{i=0}^{\min\{k,\lfloor\frac{r}{q} \rfloor \}} (-1)^i \cdot \binom{k}{i}\cdot \binom{r+k-1 -iq}{k-1}
    \end{equation*}
    Since $H_{clique}(k,r)$ is the disjoint union of $n/k$ copies of $\hat H$, thus
    \begin{equation*}
        |M_{H_{clique}(k,r),q}| = \frac{n}{k}\cdot \left( \sum_{i=0}^{\min\{k,\lfloor \frac{r}{q} \rfloor\}} (-1)^{i}\cdot  \binom{k}{i} \cdot  \binom{r + k-1 - iq}{k-1} \right).
    \end{equation*}
    Next note that $|N_{\hat H}(q)|$ is the number of solutions of the form $X_1 + \dots +X_k < r$, where $0 \leq X_i \leq q-1$ for $i \in [k]$. By the inclusion and exclusion principle, we get
    \begin{equation*}
        |N_{\hat H}(q)| = \left( \sum_{i=0}^{\min\{k, \lfloor \frac{r-1}{q} \rfloor\}} (-1)^i \cdot  \binom{k}{i} \cdot \binom{k +r-1 -iq}{k} \right).
    \end{equation*}
    Hence, the proof follows by \cref{lemma:direct-product-Fq}.
\end{proof}

As a corollary, for $q > r$, we get the following result.

\begin{corollary}
    Let $k|n$ and $r \leq k$ and $r < q$. Let $H_{clique,k,r}$ be the disjoint union of $n/k$ $\KK{k}{r}$, then
    \begin{equation*}
        |M_{H_{clique}(k,r),q}| = \binom{k+r-1}{r}\cdot \frac{n}{k}
    \end{equation*}
    and for any $d <r$,
    \begin{equation*}
        \magree{F_{H_{clique}(k,r),q}}{d} \leq \binom{k+r-1}{r-1}^{n/k}.
    \end{equation*}
    Moreover, we have
    \begin{equation*}
        \magree{F_{H_{clique}(k,r),q}}{r-1} = \binom{k+r-1}{r-1}^{n/k}.
    \end{equation*}
\end{corollary}

For the case $q\geq r$, we will show that disjoint union of cliques are not the best extremal solution to \autoref{ques:min-NH(q)}. 

\begin{lemma}\label{lemma:disjoint-turan(r+1-r)-NH-count-Fq}
    Let $r \geq 2$ be a fixed number and $r < q$. For sufficiently large $k,n$ where $k \mid n$, there exists an explicit $r$-uniform hypergraph $H_{\turan}(k,r)$ over $n$ vertices where
    \begin{equation*}
        |M_{H_{\turan}(k,r),q}| =  \lambda \cdot\frac{n}{k}\cdot \binom{k+r-1}{r} 
    \end{equation*}
    and 
    \begin{equation*}
        |N_{H_{\turan}(k,r)}(q)| \leq  \binom{k}{r}^{n/k}
    \end{equation*}
    Moreover,
    \begin{equation*}
        \lambda \leq \left( \frac{2}{1 + \lfloor \frac{r+1}{2} \rfloor} + o(1) \right).
    \end{equation*}
\end{lemma}

The proof of the above lemma can be found in \autoref{sec:proof-of-turan-Fq}.
We get the following corollary from the above lemma.

\begin{corollary}
    Let $r \geq 2$ be a fixed number and $r < q$. For sufficiently large $k,n$ where $k |n$, $F_{H_{\turan}(k,r),q}: \F_q^n \longrightarrow \F_q^{m}$, where
    \begin{equation*}
        m \leq \left( \frac{2}{1 + \lfloor \frac{r+1}{2} \rfloor} + o(1) \right)\cdot \binom{k+r-1}{r}\cdot \frac{n}{k}
    \end{equation*}
    and for any $d \leq r-1$,
    \begin{equation*}
        \magree{F_{H_{\turan}(k,r),q}}{d} \leq \binom{k}{r}^{n/k}
    \end{equation*}
\end{corollary}

Let $k$ be a large enough number and let $\lambda$ be the constant from \cref{lemma:disjoint-turan(r+1-r)-NH-count-Fq}. Let $r\geq 2$ be a fixed number such that $q > r$. If we set $a= \lceil (2\lambda)^{1/(r-1)}  \cdot (k+r-1)\rceil$, and $n$ be a large enough number such that $k \mid n$ and $a \mid n$, then by \cref{theorem:disjoint-clique-NH-count-Fq} and \cref{lemma:disjoint-turan(r+1-r)-NH-count-Fq},
\begin{align*}
    \frac{|M_{H_{\turan}(k,r),q}|}{|M_{H_{clique}(a,r),q}|}&= \lambda \cdot \frac{a}{k} \cdot \frac{\binom{k+r-1}{r}}{\binom{a+r-1}{r}}\\
    &\leq \lambda \cdot \frac{a}{k}\cdot \frac{(k+r-1)^r}{a^r}\\
    &\leq \lambda \cdot \frac{k+r-1}{k}\cdot \frac{1}{2\lambda} <1.
\end{align*}

In the following lemma, we compare $|N_{H_{\turan}(k,r)}(q)|$ and $|N_{H_{clique}(a,r)}(q)|$ and show that $|N_{H_{\turan}(k,r)}(q)| << |N_{H_{clique}(a,r)}(q)|$ for a large enough constant $r$.

\begin{lemma}
    Let $r\geq 2$ be a fixed number and $q > r$. Let $k$ be a large enough number that is growing and $\lambda \leq \left( \frac{2}{1 + \lfloor \frac{r+1}{2} \rfloor} + o(1)  \right)$ be the constant from \cref{lemma:disjoint-turan(r+1-r)-NH-count-Fq}. Let $a = \lceil (2\lambda)^{1/(r-1)} \cdot (k+r-1) \rceil $ and $n$ be a large enough number such that $k \mid n$ and $a \mid n$. Then
    \begin{equation*}
        |N_{H_{\turan}(k,r)}(q)| \leq |N_{H_{clique}(a,r)}(q)|^{\frac{r}{r-1}\cdot (2\lambda)^{1/(r-1)}\cdot (1+ o(1))}
    \end{equation*}
\end{lemma}

\begin{proof}
    For ease of notation let $H_1 = H_{\turan}(k,r)$ and $H_2 = H_{clique}(a,r)$. Then
    \begin{align*}
        \frac{\log |N_{H_1}(q)|}{\log |N_{H_2}(q)|} &\leq \frac{a}{k}\cdot \frac{\log \binom{k}{r}}{\log \binom{a+r-1}{r-1}}\\
        &\leq \frac{r}{r-1}\cdot(2\lambda)^{1/(r-1)}\cdot (1+o(1))
    \end{align*}
    which implies
    \begin{equation*}
        |N_{H_1}(q)| \leq |N_{H_{2}}(q)|^{\frac{r}{r-1}\cdot (2\lambda)^{1/(r-1)}\cdot (1+ o(1))}
    \end{equation*}
    which finishes the proof.
\end{proof}

It is not clear to us what is the optimal extremal family of $r$-uniform hypergraphs for \autoref{ques:min-NH(q)}.

We discuss the limitations of the hypergraph framework in \autoref{sec:limitations-of-hypergraph-framework}.

\bibliographystyle{alpha}
\bibliography{ref.bib}

\appendix

\section{Proof of the Interpolation Bound}
\label{sec:proof-of-interpolation-bound}

\begin{proof}[Proof of \autoref{lemma:interpolation-bound}]
    To prove the lemma we will construct $\bar{Q} = (Q_1,...,Q_m)$, a tuple of $n$-variate polynomials of degree at most $d$, that agrees on at least $M(q, n, d)$ points. We will define the formal matrix $A$, $$A \in \F_q[X_{1,1},\dots,X_{1,n},\dots,X_{M(q,n,d),1}\dots X_{M(q,n,d),n}]^{M(q,n,d) \times M(q,n,d)},$$ whose rows are indexed by $n$-variate monomials of degree at most $d$ and individual degree at most $q-1$. For any $n$-variate monomial $R$ of degree at most $d$ and individual degree at most $q-1$, and $i \in [M(q,n,d)]$ we define,
    \begin{equation*}
        A_{R,i} = R(X_{i,1},...,X_{i,n}).
    \end{equation*}
    $A$ is the formal interpolation matrix in the following sense: The evaluation of $$Q(\bar X) = \sum_{\bar e \in \{0,..,q-1\}^n \text{and } |\bar e| \leq d} q_{\bar e} \bar X^{\bar e},$$ at the points $\{ (\alpha_{i, 1}, \dots \alpha_{i, n}) : i \in [M(q, n,d)] \}$ can be expressed as by the matrix vector product $$\vec{Q} \cdot A(\alpha_{1,1}, \dots, \alpha_{1, n}, \dots, \alpha_{M(q,n,d), 1}, \dots, \alpha_{M(q,n,d), n}),$$ where $\vec{Q}$ is the coefficient vector of $Q$.\\

    By the formula of the determinant,
    \begin{equation*}
        \det(A) = \sum_{\sigma \in S(M(q,n,d))} (-1)^{sgn(\sigma)}\prod_{\substack{R\text{ is a monomial }\\ \text{of $degree$ at most d and}\\ \text{individual $degree$ at most $q-1$}  }} R(X_{\sigma(R),1} ,\dots, X_{\sigma(R),n}),
    \end{equation*}
    where $S(M(q,n,d))$ is the symmetric group on $M(q,n,d)$ elements. Since the variables are disjoint and $R's$ are distinct, none of the terms of the determinant cancels out, which implies that $\det(A) \neq 0$. 

    Moreover, $\det(A)$ has individual degree at most $q-1$. Thus by \cref{lemma:non-root-ind-deg-bound}, there exists $S \in \F_q^{n \cdot M(q,n,d)}$ of the form $S =(x_{1,1},\dots,x_{1,n},\dots, x_{M(q,n,d),1}\dots x_{M(q,n,d),n})$ such that,
    \begin{equation*}
        \det(A)(x_{1,1},\dots,x_{1,n},\dots, x_{M(q,n,d),1}\dots x_{M(q,n,d),n})\neq 0.
    \end{equation*}
    Thus, $\hat S = \{(x_{i, 1}, \dots, x_{i, n}) : i \in [M(q, n, d)]\}$ is a set of interpolation points for $n$-variate polynomials with total degree at most $d$ and individual degree at most $q-1$.
    This implies that we can interpolate $\bar{Q}$ such that for any $\alpha \in \hat S$, $\bar{Q}(\alpha) = F(\alpha)$. Hence, we are done.
\end{proof}

\section{Proof of Existence of Hard Functions}
\label{sec:proof-of-existence-of-hard-functions}

\begin{proof}[Proof of \autoref{lemma:random-function-agreement}]
    To prove this lemma we will use the probabilistic method.
    Let $\cQ$ be the set of tuples of size $m$ of $n$-variate polynomial functions of degree $d$ over $\F_q$. Observe that $\abs{\cQ} = q^{M(q,n,d)\cdot m}$, since for each of the $m$ coordinates we need to choose one of $q^{M(q,n,d)}$ many polynomials of degree at most $d$. Let $F : \F_q^n \longrightarrow \F_q^m$ be a random function where the function evaluations are chosen uniformly at random. Note that $\magree{F}{d}$ would be at least $t$ when there exists an $\bar{Q} \in \cQ$ such that $\agree(F, \bar{Q}) \geq t$. So,
    \begin{align*}
        \prob \left[ \magree{F}{d} \geq t \right] &= \prob \left[ \exists \bar{Q} \in \cQ : \agree(F, \bar{Q}) \geq t \right] \nonumber \\ 
        &\leq \sum_{\bar Q \in \cQ} \prob\left[ \agree(F, \bar{Q}) \geq t \right] & \text{(By Union Bound)} \nonumber \\
        &\leq q^{M(q,n,d)\cdot m} \cdot \binom{q^n}{t}\cdot q^{-mt}\\
        &\leq e^{M(q,n,d)\cdot m \ln q} \cdot e^{t (n\ln q  + 1) - t \ln t} \cdot e^{-mt \ln q} & \left(\binom{a}{b} \leq \left( \frac{ea}{b} \right)^b  \right)\\
        &\leq e^{M(q,n,d) \cdot m \ln q} \cdot e^{- t (m-n) \ln q } \cdot e^{-t (\ln t-1)}.
    \end{align*}

    Thus if either $t \cdot (m-n) \cdot \ln q$ or $t (\ln t-1)$ is greater than $M(q,n,d) \cdot m \cdot \ln q$, the above probability is less than $1$ (this uses the fact that $m \geq n$).
    Thus by setting
    $$t = \min \left( \frac{m}{m-n}\cdot M(q,n,d) ,\quad M(q,n,d) \cdot m \cdot \ln q \cdot \ln\left(M(q,n,d) \cdot m \cdot \ln q \right)\right),$$ we ensure that $\prob[\magree{F}{d} \geq t] < 1$, and thus by the probabilistic method there exists an $F :\F_q^n\longrightarrow \F_q^m$, such that 
    \begin{equation*}
        \magree{F}{d} < t
    \end{equation*}
    which finishes the proof.

\end{proof}

\section{Proof of the Fourier Bound}
\label{sec:proof-fourier-bound}

\begin{proof}[Proof of \autoref{lem:fourier-bound-for-far-from-linear}]
    We will apply the Fourier Transform to $F_y$ with respect to the characters of the additive group $G=\F_q^n$. For $\alpha \in \F_q^n$, $\chi_\alpha : \F_q^n \to \C$ defined as $\chi_\alpha(x) = \psi(\langle \alpha, x \rangle)$ are all the characters of $G$. Applying the Fourier Transform on $F_y$ gives us,
    \begin{align}
        \widehat F_y(\alpha)    &= \expectation_{x \in \F_q^n}[F_y(x)\chi_\alpha(x)] \nonumber \\
                            &= \expectation_{x \in F_q^n}[\psi(\langle y, F(x) \rangle) \overline{\psi(\langle \alpha, x \rangle)}]. \label{eqn:fourier-transform-for-F_y}
    \end{align}
    We will upper bound the agreement between $F$ and $L=Ax+b$ by relating it to the Fourier transform of $F_y$. For this, first observe the following orthogonality criterion for $z \in \F_q^m$,
    \begin{align}
        \indicator[z = 0] = \expectation_{y \in \F_q^m}[\psi(\langle y, z \rangle)]. \label{eqn:indicator-orthogonality-criterion}
    \end{align}
    We use the above criterion to capture the information about the agreement between $F$ and $L$ in terms of the Fourier coefficients of $F_y$,
    \begin{align*}
        \agree(F, L)    &= \sum_{x \in \F_q^n} \indicator[F(x) - L(x) =0] \\
                        &= \sum_{x \in \F_q^n} \expectation_{y \in \F_q^m} \Big[\psi(\langle y, F(x) - L(x) \rangle)\Big] &\text{( From \autoref{eqn:indicator-orthogonality-criterion} )}\\
                        &= q^n \expectation_{x \in \F_q^n} \expectation_{y \in \F_q^m} \Big[\psi(\langle y, F(x) - L(x) \rangle)\Big] \\
                        &= q^n \expectation_{x \in \F_q^n} \expectation_{y \in \F_q^m} \Big[\psi(\langle y, F(x) \rangle - \langle y, Ax+b \rangle)\Big] \\
                        &= q^n \expectation_{x \in \F_q^n} \expectation_{y \in \F_q^m} \Big[\psi(\langle y, F(x) \rangle) \overline{\psi(\langle y, Ax+b \rangle)}\Big] \\
                        &= q^n \expectation_{x \in \F_q^n} \expectation_{y \in \F_q^m} \overline{\psi(\langle y, b \rangle)} \Big[\psi(\langle y, F(x) \rangle) \overline{\psi(\langle y, Ax \rangle)}\Big] \\
                        &= q^n \expectation_{y \in \F_q^m} \overline{\psi(\langle y, b \rangle)} \expectation_{x \in \F_q^n} \Big[\psi(\langle y, F(x) \rangle) \overline{\psi(\langle y, Ax \rangle)}\Big] &\text{( Interchanging expectations )}\\
                        &= q^n \expectation_{y \in \F_q^m} \overline{\psi(\langle y, b \rangle)} \expectation_{x \in \F_q^n} \Big[\psi(\langle y, F(x) \rangle) \overline{\psi(\langle A^Ty, x \rangle)}\Big] &\text{( Since $\langle y, Ax\rangle = \langle A^Ty, x \rangle$ )}\\
                        &= q^n \expectation_{y \in \F_q^m} \overline{\psi(\langle y, b \rangle)} \expectation_{x \in \F_q^n} \Big[F_y(x) \overline{\chi_{A^Ty}(x)}\Big] &\text{( By definition of $F_y$ and $\chi_{A^Ty}$ )}\\
                        &= q^n \expectation_{y \in \F_q^m} \overline{\psi(\langle y, b \rangle)} \widehat F_y(A^Ty) &\text{( From \autoref{eqn:fourier-transform-for-F_y} )}\\
                        &= \frac{q^n}{q^m} + \frac{q^n}{q^m} \sum_{y \neq 0 \in \F_q^m} \overline{\psi(\langle y, b \rangle)} \widehat F_y(A^Ty)  \\
                        &\leq \frac{q^n}{q^m} + q^n \cdot \frac{q^m-1}{q^m} \max_{(\alpha, y) \neq (0, 0)} \abs{\widehat F_y(\alpha)} \\
                        &\leq \frac{q^n}{q^m} + q^n \max_{(\alpha, y) \neq (0, 0)} \abs{\widehat F_y(\alpha)}. \\
    \end{align*}
\end{proof}

\section{Limitations of the Hypergraph Framework}
\label{sec:limitations-of-hypergraph-framework}

\subsection{Over $\F_2$}

Now we will give a lower bound on the number of independent sets for a $r$-uniform hypergraph with $m$ edges. First, we will show the existence of a large independent set in a sparse hypergraph in the following lemma. The proof of the lemma is standard and can be found in \cite{AlonSpencer}. For the sake of completeness, we will provide the proof here.

\begin{lemma}\label{lemma:hypegraph-large-independent-set}
    Let $H =(V,E)$ be a $r$-uniform hypergraph, where $V =[n]$ and $|E| = m$, then there exists an independent set $T \subset [n]$, such that 
    \begin{equation*}
        |T| \geq \frac{n}{1 + (r-1)\cdot r\cdot \left( \frac{m}{n} \right)}.
    \end{equation*}
\end{lemma}

\begin{proof}
    Let the degree of any vertex $v \in V$ be $d_v$, then we claim that it is enough to show that there exists an independent $T \subset [n]$ such that 
    \begin{equation*}
        |T| \geq \sum_{v \in V}\frac{1}{1+(r-1)d_v}
    \end{equation*}
    since by Jensen's inequality it will follow that
    \begin{equation*}
        |T| \geq \frac{n}{1 + (r-1)\sum_{v}\left(\frac{d_v}{n}\right) } = \frac{n}{1+(r-1)r\cdot \left( \frac{m}{n}\right)}
    \end{equation*}
    which would finish the proof.\\

    Now we will show the existence of such a $T$ using the probabilistic method. We first pick a uniform random ordering of all the vertices. Next, we construct the set $I \subset V$ as follows
    \begin{equation*}
        I= \{v \in V \mid \text{ For any $u \in N(v)$}, v < u\}
    \end{equation*}
    where $N(v)$ is the neighborhood of $v$ in $H$. In particular, note that $I$ is an independent set. Next, we will compute the expected size of $|I|$ by linearity of expectation. Define indicator random variables $X_v$ for each $v \in V$ where
    \begin{equation*}
        X_v = \begin{cases}
            1 & v \in I\\
            0 & v \notin I
        \end{cases}
    \end{equation*}
    Then note that $|I| = \sum_v X_v$. Furthermore,
    \begin{equation*}
        \expectation[X_v] = \frac{1}{1+ |N(v)|} \geq \frac{1}{1+(r-1)d_v}
    \end{equation*}
    and thus by linearity of expectation, 
    \begin{equation*}
        \expectation[|I|]\geq \sum_{v \in V}\frac{1}{1 + (r-1)d_v}
    \end{equation*}
    which implies there exists an independent set $T$ such that 
    \begin{equation*}
        |T| \geq \sum_{v \in V}\frac{1}{1+(r-1)d_v}
    \end{equation*}
    which concludes the proof.
\end{proof}

As a corollary of the above lemma, we get the following lower bound for $|\Ind(H)|$.

\begin{corollary}\label{corollary:hypergraph-independent-set-lowerbound}
    Let $H = (V,E)$ be a $r$-uniform hypergraph where $V = [n]$ and $|E| = m$, then 
    \begin{equation*}
        |\Ind(H)| \geq 2^{n \cdot \left(\frac{1}{1+r^2\cdot (m/n)}\right)}
    \end{equation*}
\end{corollary}

\begin{proof}
     By \cref{lemma:hypegraph-large-independent-set}, there exists an independent set $T \subset [n]$, such that 
     \begin{equation*}
         |T| \geq \frac{n}{1 + (r-1)\cdot r\cdot \left( \frac{m}{n} \right)} \geq n \cdot \left( \frac{1}{1 + r^2 \cdot (m/n)} \right)
     \end{equation*}. 
     Since any subset of $T$ will also be an independent set, thus 
    \begin{equation*}
        |\Ind(H)| \geq 2^{|T|} \geq 2^{n \cdot \left(\frac{1}{1+r^2\cdot (m/n)}\right)}
    \end{equation*}
    which finishes the proof.
\end{proof}

\begin{remark}
    The above corollary means that for any $d$, we cannot create functions $F : \F_2^{n} \longrightarrow\F_2^{m}$ with $\magree{F}{d} < o \left(2^{n \cdot \left(\frac{1}{1+(d+1)^2\cdot (m/n)}\right)}\right)$ using our hypergraph framework.
\end{remark}

\subsection{Over $\F_q$}

Just like the $\F_2$ case, for an $r$-uniform hypergraph $H$ such that $|M_{H,q}| \leq m$, we will show a lower bound on $|N_{H}(q)|$ in the following lemma.

\begin{lemma}\label{lemma:NHq-lower-bound}
    Let $H =(V,E)$ be an $r$-uniform hypergraph, such that $|M_{H,q}|= m$, then
    \begin{equation*}
        |N_{H}(q)| \geq 2^{n \cdot \left(\frac{1}{1+r^2\cdot (m/n)}\right)}
    \end{equation*}
\end{lemma}

\begin{proof}
   Note that for any $W \in \Ind(H)$, $\indicator_{W} \in N_H(q)$, thus $|N_H(q)| \geq |\Ind(H)|$. Thus, to lower bound $|N_{H}(q)|$, it suffices to lower bound the number of independent sets of $H$. Since $|E(H)| \leq |N_H(q)|$, thus $|E(H)| \leq m$.
    By \cref{lemma:hypegraph-large-independent-set}, there exists an independent set $T \subset [n]$ such that
    \begin{equation*}
        |T| \geq \frac{n}{1 + (r-1)\cdot r\cdot \left( \frac{m}{n} \right)} \geq n \cdot \left( \frac{1}{1 + r^2 \cdot (m/n)} \right).
    \end{equation*}
    Since any subset of $T$ is also an independent set, thus
    \begin{equation*}
        |N_H(q)| \geq 2^{|T|} \geq 2^{n \cdot \left(\frac{1}{1+r^2\cdot (m/n)}\right)}
    \end{equation*}
    and we are done.
\end{proof}

\begin{remark}
    The above corollary means that for any $d,q$, we cannot create functions $F : \F_q^{n} \longrightarrow\F_q^{m}$ with $\magree{F}{d} < o \left(2^{n \cdot \left(\frac{1}{1+(d+1)^2\cdot (m/n)}\right)}\right)$ using our $(d+1)$-uniform hypergraph framework.
\end{remark}

\section{Proof of \autoref{lemma:disjoint-turan(r+1-r)-NH-count-Fq}}
\label{sec:proof-of-turan-Fq}

\begin{proof}[Proof of \autoref{lemma:disjoint-turan(r+1-r)-NH-count-Fq}]
    Let $\hat H$ be the $r$-uniform hypergraph over $k$ vertices from \cref{thm:Turan(r+1-r)-explicit-construction}. Furthermore add at most $\cO(k^{r-1})$ edges to ensure $\partial_{r-1}\hat H = \KK{k}{r-1}$. Let $\lambda = \frac{|M_{\hat H,q}|}{\binom{k+r-1}{r}}$. Let $H_{\turan,k,r}$ be $n/k$ disjoint copies of $\hat H$, then we get
    \begin{equation*}
        |M_{H_{\turan}(k,r),q}| = \lambda \cdot \frac{n}{k}\cdot \binom{k+r-1}{r}.
    \end{equation*}
    Let $\hat \lambda = \frac{|\hat H|}{\binom{k}{r}}$ be the edge density, then by \cref{thm:Turan(r+1-r)-explicit-construction}, we know that
    \begin{equation*}
        \hat\lambda \leq \left( \frac{2}{1 + \lfloor\frac{r+1}{2}\rfloor } + o(1) \right)
    \end{equation*}\
    Next we will show that $\lambda \leq \hat\lambda + o(1)$. Note that 
    \begin{align*}
        |M_{\hat H,q}| &\leq \binom{k+r-1}{r} - \binom{k}{r} + |\hat H|\\
        &\leq \hat\lambda \cdot \binom{k}{r}\cdot (1 + o(1))
    \end{align*}
    which implies
    \begin{align*}
        \lambda &\leq \hat\lambda \cdot \frac{\binom{k}{r}}{\binom{k+r-1}{r}}\cdot (1+ o(1))\\
        &\leq \hat \lambda \cdot (1 + o(1)).
    \end{align*}
    Now we will find $|N_{\hat H}(q)|$. For any $\alpha \in N_{\hat H}(q)$ such that $\sum_i \alpha_i \geq r$, since $\partial_{r-1}H = \KK{k}{r-1}$, thus it is necessary that $\alpha \in \{0,1\}^k$, which in turn implies that $\alpha = \indicator_{W}$ where $W \in \Ind(\hat H)$. Using this observation, the following is true
    \begin{equation*}
        |N_{\hat H}(q)| = \binom{k+r-1}{r-1} + \sum_{i \geq r}|\Ind(\hat H,i)|
    \end{equation*}
    where the first term is the number of number of non negative integer solutions to $X_1+ \dots +X_k \leq r-1$. Since adding edges to an $\turan(n,r,s)$ system won't increase the independence number which implies that for every $r+1$ size subset of $[k]$ is covered by an edge of $\hat H$. In particular, this gives the following lower bound on the number of edges of $\hat H$
    \begin{equation*}
        |E(\hat H)| \geq \frac{1}{r+1}\binom{k}{r} ,
    \end{equation*}
    which implies $|\Ind(\hat H,r)| \leq (1- 1/(r+1))\cdot \binom{k}{r}$,
    and thus $|N_{\hat H}(q)| \leq \binom{k}{r}$. Moreover, since $H_{\turan,k,r}$ is $n/k$  disjoint copies of $\hat H$, thus
    \begin{align*}
        |N_{H_{\turan}(k,r)}(q)| &\leq \binom{k}{r}^{n/k}
    \end{align*}
    This concludes the proof.

\end{proof}

\end{document}